\documentclass[a4paper,11pt]{article}
\usepackage{graphicx}
\usepackage{fullpage}
\usepackage{amsthm}
\newtheorem{lemma}{Lemma}[section]
\newtheorem{theorem}[lemma]{Theorem}
\newtheorem{definition}[lemma]{Definition}
\newtheorem{algorithm}{Algorithm}

\newcommand{\CiPM}{\textit{C1PM}}

\begin{document} 

\title{\Large Linear-Time Recognition of Probe Interval Graphs}

\author{Ross M. McConnell\thanks{Computer Science Department, Colorado State University, Fort Collins, CO
80528, USA, \texttt{rmm@cs.colostate.edu}} \and Yahav Nussbaum\thanks{The Blavatnik School of Computer Science, 
Tel Aviv University, 69978 Tel Aviv, Israel, 
\texttt{yahav.nussbaum@cs.tau.ac.il}}}

\date{}

\maketitle

\begin{abstract}
The interval graph for a set of intervals on a line consists of one vertex 
for each interval, and an edge for each intersecting pair of intervals. A 
probe interval graph is a variant that is motivated by an application to 
genomics, where the intervals are partitioned into two sets:  probes and 
non-probes. The graph has an edge between two vertices if they intersect 
and at least one of them is a probe.  We give a linear-time algorithm for 
determining whether a given graph and partition of vertices into probes 
and non-probes is a probe interval graph.  If it is, we give a layout of 
intervals that proves this.  We can also determine whether the layout
of the intervals is uniquely constrained within the same time bound.
As part of the algorithm, we solve the consecutive-ones
probe matrix problem in linear time, develop algorithms for operating
on PQ trees, and give results that relate PQ trees for different submatrices
of a consecutive-ones matrix.
\end{abstract}

\section{Introduction}
\label{sect:intro}

The {\em intersection graph} for a collection of sets has one vertex
for each of the sets and an edge between two vertices if the corresponding
sets intersect.
An {\em interval graph} is the intersection graph of a set of intervals 
on a line.  The set of intervals constitutes an {\em interval model} of 
the graph.  Figure~\ref{fig:intvalGraph} gives an example.
Interval graphs play an important role in many problems,
see \cite{CLRS09,Gol80,GolT04}. 

\begin{figure}
\centerline{\includegraphics[]{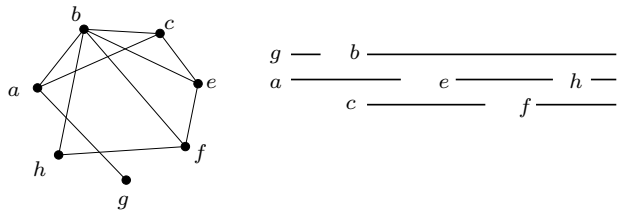}}
\caption{An interval graph and an interval model of it.}\label{fig:intvalGraph}
\end{figure}

The problem of recognizing whether a 
graph is an interval graph played a key role in the 1950's in establishing
the linear topology the fine-scale organization of genetic information
in DNA~\cite{Ben59}.   The linear topology of a DNA molecule had been known since 1953,
when it was described by Watson and Crick.  In addition, it was known that the collection of genes
had a linear arrangement along the chromosome, since this arrangement could be inferred from
recombination frequencies of alleles.

This was not enough to establish what we now know, which is that the 
genetic information in a chromosome is written in its entirety onto a single DNA molecule (actually,
two identical DNA molecules called sister chromatids).  What was known at that time
did not exclude the possibility that the fine structure of the chromosome
was organized around multiple independent DNA molecules or one with small branches,
giving it a tree-like topology that only appeared to be linear on a large scale.

To test the hypothesis that the fine structure was linear, Seymour Benzer
isolated 145 mutant strains of a bacteria-infecting virus, T4~\cite{Ben59}.
He further hypothesized that each mutation occupied a contiguous region
of the genome, which would be an interval of the genome if the topology
was indeed linear.  The test of these hypotheses consisted of finding a method
for determining which pairs of mutations occupied intersecting regions of the genome,
and then testing whether the derived intersection graph was an interval graph.

By themselves, the strains are not viable.
When bacteria are infected with two of the strains, however, the viruses can recombine
their genomes to assemble the original viral genome, giving rise to
viable viruses, provided that the regions occupied by their two mutations do not intersect.
By infecting bacteria with a pair of strains and determining whether viable viruses
arose, he was able to deduce the presence or absence of an edge in the intersection
graph with high accuracy.  He found an interval graph on 144 of the 145 strains
that was consistent with thousands of tests.  The anomalous strain had to be excluded when
it was found to have two mutations that were not contiguous.

The fraction of all graphs that are interval graphs on 144 vertices
is minuscule, so this was a strong test of the hypothesis.  The interval model
he found for the graph gave one possible linear arrangement of the
intervals occupied by the mutations, but this was not uniquely constrained by
the graph, hence by the data.

His paper gives the first characterization of interval graphs, based on combinatorial
properties of their adjacency matrices, together with a heuristic
for recognizing whether a graph is an interval graph.
According to an acknowledgment at the end of the the paper, 
the characterization of interval graphs and the heuristic he used for recognizing them 
was due to the prominent biochemist Leslie Orgel of Cambridge, who suggested it to him in a 
personal communication.  As far as we know, the original characterization
has not previously been attributed to Orgel in the graph theory literature.

This ignited considerable interest in the combinatorial study of interval graphs.
Lekkerkerker and Boland gave a characterization in terms of forbidden induced subgraphs 
in 1962~\cite{LB62}.
It also sparked a search for efficient
algorithms for determining whether a graph is an interval graph~\cite{FulkGross65}
and producing an interval model if it is.
Booth and Lueker gave the first linear-time algorithm for 
the problem in the 1970's~\cite{BL76}.  Their algorithm determines whether
the graph uniquely constrains the linear arrangement of the intervals.

A {\em (partitioned) probe interval graph}~\cite{McMorris98}
(also called interval probe graph) is a graph in 
which the vertex set is partitioned into {\em probes} and {\em non-probes}. 
It is a generalization of intersection graph of an interval model, such that 
the graph has an edge between two vertices if their intervals intersect {\em and 
at least one of them is a probe}. Information about intersections of non-probe
intervals is missing.  Such a model is a {\em probe interval model} of the graph.

There has been quite a bit of work on topological and combinatorial properties of 
these graphs; see~\cite{GolT04} for a survey.  
The motivation that initially gave rise to interest in the class
was for physical mapping of genomes~\cite{ZhangPatent,Zhang94, Zhang94etal}.
Rather than using mutations to deduce intersections, small fragments
of single-stranded DNA from a region of a genome were
cloned and embedded in a filter, and then tested against probes taken from
the complementary strands to see which 
probes hybridized (bonded) with them.  Hybridization indicates that the
strands share a section that encodes for the same sequence of base pairs.
If this sequence is long enough, it occurs on a unique interval of the genome,
and therefore indicates that the strands come from intersecting intervals.
Only a subset of the fragments were used as probes.
The intersections inferred by the procedure were represented by a probe
interval graph, since no direct information about hybridization between
non-probes was available in the experimental data.  The possibility of deducing a probe 
interval model from the graph, especially if it
was uniquely constrained, provided a possible way to infer the linear arrangement of the
fragments.  

At the time, no efficient algorithm for deducing a probe interval model was known.
The possibility that one might not exist was recognized as an obstacle 
to the approach~\cite{McMorris98}.  
In~\cite{Zhang94etal}, a heuristic based on breadth-first search was instead
applied to the probe interval graph in the hopes of finding a chordless path in the
probe interval graph that consisted of intervals that spanned the region of interest.
An efficient algorithm for constructing probe interval models, had it been available, could
have solved this problem reliably.

A polynomial algorithm for the recognition and construction of probe interval models
was first given by
Johnson and Spinrad~\cite{JohnsonSpin01}, who gave an $O(V^2)$ bound.
Using a different approach, McConnell and Spinrad gave an $O(V+E \log V)$ 
algorithm~\cite{McCSpin02}.  Uehara claimed an $O(V^2+VE)$
algorithm~\cite{U04} that checks whether the model is unique, though
some of the details have not been fully described.  

In this paper, we give the first linear-time algorithm for the problem.
That is, the input is a graph $G=(P,N,E)$, where the $\{P,N\}$ is a partition
of the vertices into probes and non-probes.  The algorithm determines whether there
exists a probe interval model of $G$ consistent with this partition of the vertices
into probes and non-probes.  If there is, the algorithm returns such a model.
The algorithm also
determines whether the layout of the model is unique, according to a slight
generalization of a definition of uniqueness that is well-known in the case 
of interval graphs, and suitable for deducing the arrangement of intervals in a genome.
This work appeared in preliminary form in~\cite{McCNuss09}.  

All of the above algorithms take as input a graph whose
vertex set is partitioned into probes and non-probes.
Chang et al.~\cite{CKLP05} consider the problem of recognizing this graph
class when this partition is not given.

The original motivating application to biology has been superseded by more reliable
and economical techniques.  However, hybridizing 
pairs continue to be of interest for inferring whether
intervals occupied by DNA fragments intersect~\cite{Fullwood09}, giving
rise to intersection graphs.  

In some cases, the intersection graphs continue to be
modeled by probe interval graphs, rather than full interval graphs.
For example, in
{\em contig scaffolding}, there are two kinds of segments of the genome, {\em contigs} and 
intervals with {\em paired-end tags}~\cite{scarpa13}.  The sequence represented by a contig is known.
Only small sequences at the extreme ends of a tagged intervals are known, however; these 
are the paired-end tags.  The intersections of contigs
with each other and with paired-end tag intervals can be deduced.  
The tags 
are small compared to the lengths of their intervals.   Therefore,
when two of these intervals
intersect, their end tags are unlikely to intersect, and the intersection of their
intervals cannot be inferred directly.
The graph given by the hybridization data is therefore modeled by a probe interval graph, where
the contigs are the probes and the paired-end tag intervals are the non-probes.

Our algorithm would probably have
been more useful for Benzer than Booth and Lueker's, had these algorithms been available at the time,
because he did not have the resources to perform all
$n(n-1)/2$ experiments needed to construct the full interval graph on $n=144$ mutations.

An obstacle to the usefulness of efficient algorithms for  construction of interval models from 
hybridization data is that, at the time of this writing, hybridization data are
prone to many false positives and false negatives.
This corrupts the experimentally derived graph by adding or removing edges
in the true intersection graph, and with high probability, this will give
rise to a graph that is no longer an interval graph or probe interval graph.  
This
is one reason that the proposed method of~\cite{Zhang94,ZhangPatent} never
become competitive with alternative sequencing techniques.  Advances in
in the accuracy of detecting
intersections, either by hybridization or by some unforeseen method,
before our algorithm for constructing probe
interval models, or Booth and Lueker's algorithm for constructing interval
models, are very useful in physical mapping or sequencing.  Benzer's methods, however, illustrate
that it may be difficult to foresee clever future laboratory methods that could give
rise to highly accurate intersection data in the future.

Despite impracticality of applying it to noisy biological data, 
the algorithm of Booth and Lueker continues to be studied 
by people working with hybridization data, because of the structural insights it gives
about the graph class.  The concepts it uses, such as the {\em consecutive-ones property}
and {\em PQ trees} (discussed below) have given
rise to methods that are more tolerant of errors in the data~\cite{Pevzner00}.  
One contribution of our paper is to give analogous structural insights into the class of probe
interval graphs.

A {\em circular-arc graph}
is the intersection graph of arcs on a circle.  Circular-arc graphs are a generalization
of interval graphs; interval graphs are those circular-arc graphs that have
a model where the arcs do not cover the entire circle.  
This generalization, which reflects constraints in cyclic scheduling problems, for
example, is much less structured than interval graphs.
For example, in interval graphs, the number of maximal cliques is bounded
by the number of vertices, while in circular-arc graphs, it can be exponential
in the size of the graph~\cite{Tucker80}.  Interval graphs
are a subclass of the class of perfect graphs, while circular-arc graphs are 
not.  When Booth and Lueker developed their linear-time algorithm for recognizing
interval graphs and producing interval models, Booth conjectured that the corresponding
problems for circular-arc graphs would 
turn out to be NP-complete~\cite{Booth75}.
The conjecture was later disproved~\cite{Tucker80}, and
the first linear-time algorithm was given
in~\cite{McCCircPaper}.  

The ability to find a probe interval model for
a probe interval graph was an essential step in the result of~\cite{McCCircPaper}.
The class of probe interval graphs was described independently in an early draft of that
paper, before it was pointed out that it had previously been described in connection
with a biological application.
The paper used an algorithm for recognizing probe interval graphs and
finding a probe interval models that is described in~\cite{McCSpin02}.
Even though the algorithm of~\cite{McCSpin02} is not linear,
it does not violate the linear time bound for the circular-arc graph recognition
algorithm, since it operates on a graph whose size is sublinear in the size of $G$.

One interpretation of a probe interval graph is that it is an interval graph 
where reports of adjacencies by non-probe vertices are missing or distrusted.
Adjacencies between the trusted and the untrusted vertices can be obtained from the 
trusted vertices; only adjacencies between pairs of untrusted vertices 
are unknown.  This interpretation is driven by applications that have been identified
so far.  

Another interpretation is that interval graphs are used to represent conflicts and compatibilities in
scheduling problems.  Their membership in the class of
perfect graphs (see \cite{Gol80}) gives rise to efficient
algorithms for finding maximum independent sets, maximum cliques, and minimum colorings~\cite{CLRS09},
which correspond to sets of interest for finding efficient schedules.
Probe interval graphs, which are also perfect graphs (see~\cite{McMorris98}), introduce a third possibility,
which is that a subset of the jobs do not have a conflict even if
their intervals intersect.  For example, if some of
the jobs require a dedicated resource for technical or security reasons,
while others can share the resource,
then the conflicts are modeled with a probe interval graph, where the jobs
that require exclusive access are the probes and those that do not are the 
non-probes.  
A linear-time algorithm
for finding a maximum clique and minimum coloring, given a probe
interval model, is given in~\cite{GolT04}.  A consequence of our
result is therefore a linear-time algorithm for minimum coloring and 
maximum clique, given the partitioned probe interval graph.
An open problem is whether a maximum independent set and
minimum clique cover can be found in linear time from the probe interval
model.

A {\em consecutive-ones ordered matrix} is a 0-1 matrix in which 
1's in each row are consecutive.  A {\em consecutive-ones ordering}
of a 0-1 matrix is permutation of its columns that gives a
consecutive-ones ordered matrix.
A 0-1 matrix has the
{\em consecutive-ones property} if there exists a consecutive-ones 
ordering of it.   A family ${\cal F}$ of subsets of a set $C$ has the {\em consecutive-ones
property} if there exists an ordering of elements of $C$ such that each
member of ${\cal F}$ is consecutive.  The two concepts are equivalent, since
${\cal F}$ can be represented using one row for each $X_i \in {\cal F}$, one
column for each element of $C$, and a 1 in row $i$, column $j$ if
set $i$ contains element $c_j$ of $C$, and the resulting matrix has
the consecutive-ones property if and only if the set family does.
A {\em consecutive-ones matrix} is one that has the consecutive-ones property;
it is not necessarily consecutive-ones ordered.

As part of their algorithm to recognize
interval graphs, Booth and Lueker developed an algorithm to determine
whether a matrix is a consecutive-ones matrix, and, if so, to produce
a consecutive-ones ordering of it, in time proportional to the number
of rows, columns, and 1's in the matrix, given a sparse representation.   
We use this result extensively in this paper.

The {\em consecutive-ones sandwich}
problem is an extension of the consecutive-ones problem, where each entry is 0, 1 or $\ast$. 
An $\ast$ is a ``don't care''; it can stand for either a 0 or a 1. The problem is 
to find an assignment of 0's and 1's to the $\ast$'s such
that the resulting matrix has the consecutive-ones property.  
Deciding whether this
is possible is NP-Complete~\cite{GolumbicSandwich98}.  
This fact was recognized
as a possible obstacle to efficient construction of probe interval models in~\cite{McMorris98}.
The {\em consecutive-ones probe matrix} problem is the special case where we require that the $\ast$'s 
form a submatrix (see also \cite{CGKN07}).  This is also a generalization of the consecutive-ones problem.
We give an algorithm that takes time that is linear in the number of rows, columns, and 1's in $M$
to find a solution or determine that none exists.  This requires an efficient representation
of $M$ that does not represent the $\ast$'s explicitly, and our solution gives an
implicit assignment of 0's and 1's to $\ast$'s.  The $\ast$'s can be assigned
values explicitly, but the number of them might not be linear
in the size of the input.  In this paper, we develop methods for reducing the problem of constructing
probe interval models to that of solving the consecutive-ones probe matrix problem.

When Booth and Lueker's algorithm determines that a matrix has
the consecutive-ones property, it gives an implicit representation of
{\em all} consecutive-ones orderings of a matrix, called a {\em PQ tree}. 
(See Figure~\ref{fig:PQExample}.)

The leaves of the PQ tree are the one-element subsets of the set of columns of the matrix.
The PQ tree gives all consecutive-ones 
orderings by constraining the orderings of children of internal nodes, as follows.  Some of 
the internal nodes are labeled {\em P nodes}.   For such a node every ordering of
its children is permitted.  Others are labeled {\em Q nodes}.   For such a node, an 
ordering of its children is given;
the only permissible orderings of its children are the given ordering
and its reverse. For a PQ tree $T$, let $\Pi(T)$ 
denote the set of all possible orderings of its leaves, given these 
constraints.  
$T$ uniquely determines $\Pi(T)$.
We consider different orderings of the PQ tree that are consistent with the
constraints to be the same PQ tree.
The algorithm of Booth and Lueker~\cite{BL76} either finds the unique
PQ tree $T(M)$ such that $\Pi(T)$ is equal to the consecutive-ones 
orderings of columns of a matrix $M$, or determines that the matrix is not 
a consecutive-ones matrix. 
It is also easy to see that $\Pi(T)$ uniquely determines $T$ for any PQ tree. 
One way to see this is an algorithm we give below that, given $T$, constructs a matrix
whose consecutive-ones orderings are $\Pi(T)$, which then has a unique PQ tree by
Booth and Lueker's result.

\begin{figure}
\centerline{\includegraphics[]{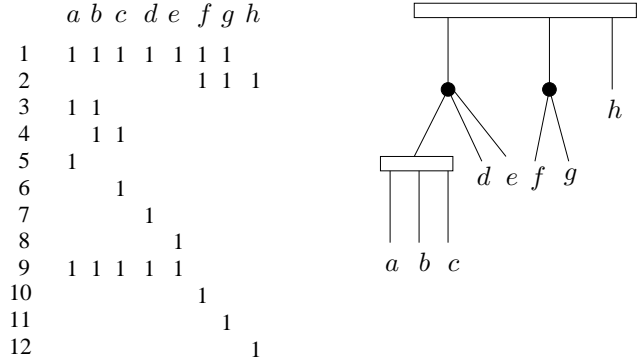}}
\caption{The PQ tree of a consecutive-ones matrix.
Only the 1's are depicted in the matrix; other entries are 0's.
The leaves of the tree are the columns of the matrix.
The P nodes are represented with black discs and the Q nodes are represented
with rectangles.  At each P node, the children can be ordered arbitrarily,
and at each Q node, the children can be ordered in the depicted way or
its reverse.   The resulting ordering of leaves is always a
consecutive-ones ordering
of the columns of the matrix.  All consecutive-ones orderings of
the matrix can be obtained in this way.}\label{fig:PQExample}
\end{figure}

A significant part of our paper is devoted to developing general results
about PQ trees and consecutive-ones ordered matrices that then allow us
to derive our algorithm for probe interval graphs.  
We develop proof techniques and useful results about
the relationships between the PQ trees of a matrix and those of
its submatrices.   See, for example, Section~\ref{sect:PQtrees}, and, for examples of
applications, Sections~\ref{sect:sufficiency} and~\ref{sect:C1PM}).
We give examples of how
Booth and Lueker's algorithm can be exploited as a black box for answering
constraint satisfaction questions 
that do not correspond to ones that it had been
originally designed for.  See, for example, Sections~\ref{sect:NP-P} and~\ref{sect:PP}. 

Uehara claims a data structure for implicitly representing all possible probe
interval models, though some details required to verify it are missing~\cite{U04}.
The time bound he gives for constructing it is $O(V^2 + VE)$.
We develop a structure based on a pair of PQ trees that has this capability
(Figure~\ref{fig:probeOnes}), and we construct it in time that is linear in the
size of the graph.

\section{Preliminaries}\label{sect:prelim}

\subsection{Notation}

Except for some additional definitions, we use standard terminology
and conventions from~\cite{CLRS09}.
For example, that text states that the space
requirement of the adjacency-list representation of a graph is $\Theta(n+m)$,
since it requires that many integers and pointers (words of memory in the RAM model),
even though the number of bits required is $\Theta((n+m)\log n)$.  We use
the RAM model in this paper for measuring space requirements, not just time requirements.

Given a graph, we let $n$ denote the number of vertices and $m$ the
number of edges.
We will assume the standard adjacency-list representation of a graph.
Let $N(v)$ denote the {\em open neighborhood} of $v$, that is, the set 
of neighbors of $v$ in $G$, and let $N[v]$ denote its {\em closed 
neighborhood}, that is, $\{v\} \cup N(v)$.

By $G=(P,N,E)$, we denote a probe interval graph with probes $P$
and non-probes $N$.  Let $V = P \cup N$ denote the vertex set.
If $X$ is a nonempty subset of $P \cup N$, let $G[X]$ denote the subgraph 
of $G$ induced by $X$, together with the 
classification of members of $X$ as probes or non-probes.

If $X$ is a set, let $X - c$ denote $X \setminus \{c\}$.  
If ${\cal R}$ is a collection of sets, let ${\cal R} - c$
denote $\{X - c \mid X \in {\cal R}\}$.  If $G=(V,E)$
is a graph and $u$ is a vertex, let $G - u$ denote $G[V - u]$.
More generally, if $U \subset V$, let $G - U$ denote $G[V \setminus U]$.

If $M$ is a 0-1 matrix, let $R(M)$ and $C(M)$ denote its rows and columns,
respectively.  If $Y$ is a subset of its rows, and $X$ is a subset of its columns,
$M[Y][X]$ denotes the submatrix given by rows of $Y$ and columns of $X$.
When we wish to restrict only the row set, we denote this $M[Y]$; it is implied
that columns are $C(M)$.  When we wish to restrict only
the column set, we denote this $M[][X]$; it is implied that the row
set is the $R(M)$.  If $c$ is a column of $M$, then we let $M - c$
denote $M[][R(M) \setminus \{c\}]$.  

We will often treat the columns of a 0-1 matrix as {\em sets}, where 
each column is the set $R$ of rows in which the column has a 1.
A shortcoming of this convention is that, unlike a dynamic list,
a set has no identity independent of its contents.  
We want a column to retain its identity when we add or remove a row from
the matrix, even though the set it represents may change.
Also, if two columns represent the same set, we want
them to have separate identities, which they retain even when we permute the 
column order.

We therefore assume that each column $x$ has
an identity separate from its current contents,
much like a dynamic list.  
Taking a submatrix $M[Y]$ of $M$ can be seen as an operation on column lists.
This allows us to say that $M[Y]$ and $M$ have 
different sets of rows but have the same column set.  
If $x$ is a column of $M$, we let $S(M,c)$ denote the set of rows where the
column has 1's.  Note that $S(M[Y],c) = S(M,c) \cap Y$.
If $M$ is understood, we let $S(c)$ denote $S(M,c)$.  
Rows are handled in a symmetric way; if $r$ is a row, $S(M,r)$ denotes the set
of columns of $M$ where the row has a 1, and if $X$ is a set of columns of $M$,
$S(M[X],r) = S(M,r) \cap X$.  Though this notation is convenient
in mathematical expressions, we will sometimes ignore the distinction between
a column and the set it represents in English sentences
when the meaning is clear.  
For example, we can say
that column $c$ is a {\em subset of another column} instead of 
the more formal $S(c) \subseteq S(c')$ for some $c' \in C(M)$ such that $c' \neq c$.

A sparse representation of a binary matrix can be obtained by giving to each
1 a pointer to the next and preceding 1 in its row and the
next and preceding 1 in its column.
The size of the representation, as measured on the RAM model,
is proportional to the number of rows,
columns and 1's and we consider an algorithm to run in linear time if it 
runs in time proportional to this size.
Using elementary methods, such a representation can be obtained in linear 
time from a list of the positions of nonzero elements in each row or in 
each column.  The order of rows and columns can also
be permuted arbitrarily in linear time.
A submatrix can be represented with ordered lists of pointers to
a subset of rows and columns.  

\subsection{Classes of graphs and matrices}
\label{sect:classes}

We define the {\em cliques} of a graph to be its {\em maximal} complete 
subgraphs.  We assume that the vertices of a graph are numbered from 1
through $n$.
A {\em clique matrix} of a graph is a 0-1 matrix with one row for each
vertex, one column for each clique, and a 1 in row $i$, column $j$
if vertex $i$ is a member of clique $j$.
In this paper, we consider
two clique matrices to be equal if and only if they are
equal in the standard sense of matrix equality in linear algebra.  This differs from
some papers that refer to {\em the} clique matrix, reflecting the
view that the purpose of the matrix is to represent the family of cliques,
and the order of columns is unimportant.  By our convention, an interval
graph with $k$ cliques has $k!$ clique matrices, all of which have
the consecutive-ones property but not all of which are consecutive-ones
ordered.

A vertex $v$ is {\em simplicial} if $N[v]$ is a complete subgraph;
in this case, $N[v]$ must be a clique.

A {\em chordal graph} is a graph with no induced cycle of size greater 
than three. A chordal graph has 
$O(n)$ cliques, and the sum of their cardinalities is $O(n+m)$,
so a sparse representation of a clique matrix takes $O(n+m)$ space.
It takes $O(n+m)$ time to find a sparse representation of a
clique matrix of a chordal graph by the algorithm
of Rose, Tarjan and Lueker~\cite{RTL:triangulated}.

Booth and Lueker's algorithm~\cite{BL76} for recognizing interval graphs 
uses the algorithm of Rose, Tarjan and Lueker either to determine
that the graph is not chordal, in which case it is not an interval graph,
or else to produce a sparse representation of a clique matrix.  

It then
uses the fact that a graph is an interval graph if and only if
its cliques have the consecutive-ones property.  
The central element of their recognition algorithm is an algorithm
for either finding a consecutive-ones ordering of a 0-1 matrix, or
else determining that none exists.  They apply this to a clique
matrix of the chordal graph to determine whether it is an interval
graph.  Figure~\ref{fig:intvlExample} gives an example.

\begin{figure}
\centerline{\includegraphics[]{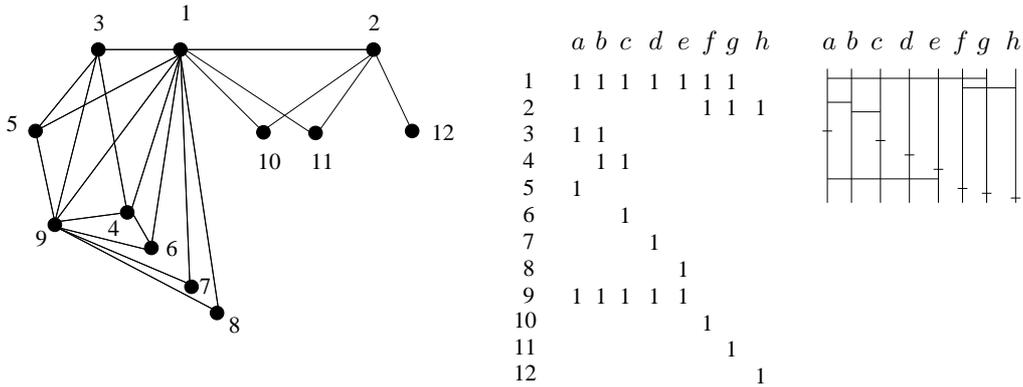}}
\caption{An interval graph, a consecutive-ones ordered-clique matrix
(the matrix of Figure~\ref{fig:PQExample})
and a schematic representation of it.}\label{fig:intvlExample}
\end{figure}

To see why a graph is an interval graph if its clique matrices are
consecutive-ones matrices, note
that the consecutive-ones ordering of a clique matrix defines
an interval model:  the interval for each vertex extends
from the first to the last column of the block of consecutive ones
in its row.  
Two vertices of a graph are adjacent if and only if
they are members of a common clique, so two of these intervals
intersect if and only if their vertices are adjacent.
Thus Booth and Lueker's algorithm produces an interval model
whenever the input graph is an interval graph.

To see why the clique matrices of every interval graph have the consecutive-ones
property, let $G$ be an interval graph.  There exists a set ${\cal I}$ of intervals 
on the line, one for each vertex, whose intersections give the edges of $G$.
For each clique $K$ of $G$, the intervals corresponding to $K$ are pairwise
intersecting.  Any set of pairwise intersecting intervals must have
an intersection point $p$ in common; this is known as the {\em Helly
property.}  Associating
one such point on the line for each clique gives a left-to-right  ordering of the cliques.
For each vertex, the cliques that contain it are those whose
associated points lie in the vertex's interval.  These cliques are consecutive in
the left-to-right ordering, so this ordering of the cliques is a consecutive-ones
ordering.  For example, in Figure~\ref{fig:intvalGraph}, the cliques from
left-to-right are $(\{a,g\}, \{a,b,c\}, \{b,c,d\}, \{b,e,f\}, \{b,f,h\})$,
and for each vertex, the cliques that contain the vertex are consecutive
in this ordering.

An important fact for our purposes is that Booth and Lueker's algorithm 
for finding a consecutive-ones ordering can operate on an arbitrary 0-1 matrix,
not just a clique matrix of a chordal graph.  Its input is a sparse representation
of the matrix, and it takes time proportional to the number of rows, columns, and 1's
in the matrix.  
If a 0-1 matrix is consecutive-ones ordered, let the {\em left endpoint} of a 
row be the column of its first 1, the {\em right endpoint} be the column of 
its last 1, and the row's {\em interval} the block of columns where it has 1's.

For the consecutive-ones probe matrix problem, we seek to represent the inputs in space proportional 
to the number of rows, columns, and $1$'s in the probe matrix.  In other words, 
we do not have to represent the $\ast$'s explicitly.  
Let $M_R$ be the submatrix of the probe matrix $M$ 
whose rows are the rows that do not have $\ast$'s, and whose columns are 
all columns of $M$.  Let $M_C$ be the submatrix of $M$ whose columns are 
the columns that do not have $\ast$'s, and whose rows are all rows of 
$M$.  The columns of $M_C$ are a subset of the columns
of $M_R$, and the rows of $M_R$ are a subset of the rows
of $M_C$.
We represent the input to the problem using sparse representations of $M_R$ and $M_C$,
neither of which contain $\ast$'s.

A solution is any a consecutive-ones ordering of columns of $M_R$, hence
of the columns of $M$,
such that the subsequence given by columns in $M_C$ is also a consecutive-ones
ordering of $M_C$.  This assigns a position to each column of $M_C$ among columns of $M$.
An $\ast$ is implicitly a 1
if it occurs between two 1's from columns of $M_C$ in the ordering of columns of $M$.
Since $M_C$ is consecutive-ones ordered, this gives a consecutive-ones ordered matrix.
We solve the problem in time linear in the number of rows, columns,
and 1's in $M_R$ and $M_C$.
This time bound does not allow us to explicitly assign 1's to the $\ast$'s;
they are implied by the column order.  This nevertheless allows linear-time
construction of a simple data structure that allows $O(1)$ lookup of the value
in any row and column of $M$, by storing the column number of the first and the last 
1 of every row.

\section{Overview}

At various points,
our algorithm may find that a required property is not met, when it must hold if $G$
is a probe interval graph.  In this case, we may reject $G$.  For instance,
if $G$ is a probe interval graph, then the subgraph $G[P]$ induced by
the probes $P$ is an interval graph, so a clique matrix of this subgraph must have
the consecutive-ones property.  An initial step is to find a consecutive-ones
ordering of this matrix, which is required in order to carry out the next
steps of the algorithm.  If does not have one,
then we can reject $G$ and halt.
We therefore assume that $G$ is a probe interval graph and use this to prove properties that are
required at each step.  If we cannot perform the operations
as described at the step because the properties do not hold, we reject $G$ and halt.
The algorithm also tests for required properties before a step if their absence could undermine
the time bound before the problem is noticed.  If it does not reject $G$,
the algorithm constructs a probe interval model of $G$.

The notion of a model of an interval graph generalizes easily to probe interval graphs.
Henceforth in this paper, we use the following convention:

\begin{definition}\label{def:model}
A {\em probe interval model} of
a probe interval graph $G=(P,N,E)$ is a consecutive-ones-ordered matrix that has one
row for each vertex, such that two vertices are neighbors in $G$
if and only if their rows intersect and at least one of the vertices is a probe.
An {\em interval model} of an interval graph is a probe interval model
that has no non-probes.
\end{definition}

The consecutive-ones ordered clique matrices of an interval graph are not the only models
of an interval graph that satisfy Definition~\ref{def:model}.  Between any two
consecutive cliques $c_i$ and $c_{i+1}$, it
is easy to see that a column $c$ can be added such that $S(c)$ is a subset
of $S(c_i)$ or of $S(c_{i+1})$ and 
supersets of $S(c_i) \cap S(c_{i+1})$ without affecting the represented graph.  

Notice that the probes in each column induce a complete subgraph in $G$, but the
vertices in a column do not induce a complete subgraph if the column contains
more than one non-probe, since non-probes are nonadjacent.  
The Helly property requires that a each clique of $G[P]$ be a subset of
some column in every probe interval model.

Booth and Lueker treat consecutive-ones ordered clique matrices as the
``canonical'' or ``normal'' representation of an interval model.
Let us call this an interval model in {\em normal form}.
Restricting the focus to this normal form has distinct advantages.
This representation distills down the information that the graph gives 
about possible structures of models, without representing arbitrary details
that cannot be deduced from $G$.
Every one of these models is a consecutive-ones ordering of a single model,
allowing the PQ tree to give a representation of all of them.  They gives a precise
definition to what it means for the model to be uniquely constrained:  the model
is unique, up to reversal of columns.  Also, the number of 1's in the clique
matrix is $O(n+m)$, and this is not true for arbitrary models.  Models that are
not clique matrices are implicitly represented by those that are.

We therefore also seek a generalization of this standard form to probe interval models.
If $M$ is a probe interval model, let the {\em probe set} in a column denote
the probes that are members of the column.
A {\em contraction} of a row in a probe interval model is the operation
of changing its first or last 1 to a 0, resulting in a shorter interval
for the row's vertex.
A row is {\em taut} if contracting it changes the
represented neighborhood of the row's vertex.  A probe
interval model is {\em taut} if every row is taut.
Two consecutive columns in a model can be {\em merged} if they
can be replaced with their union without changing the represented probe interval graph.
A model is {\em minimal} if no two consecutive columns can be merged.
Two such columns can be
merged if one is a subset of the other, or if they have the same
probe set, since making non-probes subsets of a common column
does not represent them as adjacent.

\begin{definition}\label{def:normal}
A probe interval model is a {\em normal probe interval model} 
if it is taut and minimal.
\end{definition}

It is easy to see that in the special case of an interval model, where
there are only probes, a model is
a normal model if and only if it is a consecutive-ones ordered clique matrix.
Therefore, a normal probe interval model is a generalization of the 
an interval model in normal form.
We show below that, just as in the case of interval graphs, every probe interval 
graph has a normal model.  As in the case of clique matrices, a normal model has 
$O(n+m)$ 1's.  This is important for the time bound.
Note that $G[P]$ is an interval graph.
Each clique of $G[P]$ is a subset of exactly one column
of every normal model of $G$ (\ref{lem:cliquesUnique}), just as in the special case of normal
interval models.  If a probe interval graph has no simplicial non-probes, the normal 
models are the consecutive-ones ordering of a single normal model (\ref{thm:allModels}), 
and the PQ tree of this model therefore gives a representation of all normal models,
just as in the case of normal interval models.  

Unfortunately, this is not true of probe interval graphs that have simplicial non-probes.
The difficulties posed by simplicial non-probes were previously identified by Uehara~\cite{U04}.
That a PQ tree does not suffice to represent all probe interval models of a partitioned
graph is illustrated by Figure~\ref{fig:badSimplicials}.

\begin{figure}
\centerline{\includegraphics[]{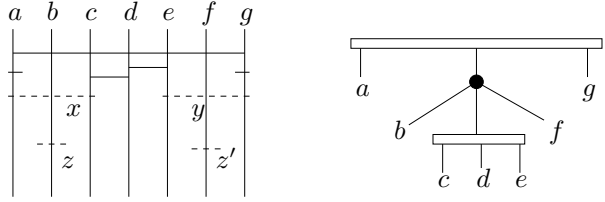}}
\caption{A PQ tree cannot represent all possible arrangements of intervals in a probe interval model
that contains both simplicial and non-simplicial non-probes.
On the left is a probe interval model, where the dashed lines are the
non-probes.  Columns $b$ and $f$ owe their
existence to simplicial non-probes $z$ and $z'$.  The positions of $z$ and $z'$ can
be swapped without otherwise changing the order of the columns to obtain
a new model for the same graph.  This suggests
the PQ tree at the right.  However, now there is no way for the PQ tree
to reflect the constraint, imposed by the other intervals in the model,
that the relative order of $a$ and $g$ constrains the order of $(c,d,e)$.  
The orderings expressed by a PQ tree have a type of ``context-free'' property, 
in the sense that the orderings expressed by a subtree are independent of any larger
context, and this is not sufficient for expressing all probe interval
models of a graph that has both simplicial and non-simplicial non-probes.
}\label{fig:badSimplicials}
\end{figure}

We give an algorithm for finding whether a probe interval graph has a unique
normal model, up to reversal of column order.

Every column of a normal model contains an endpoint of a row.
Let the {\em clique columns} be those that contain a clique of $G[P]$.
We show that a  column is a clique column if and only if it contains both left and right 
endpoints of members of $P$.
Let $N_S$ be the simplicial non-probes.
Let the {\em simplicial columns} be those that are not clique columns and that
only contain endpoints of members of $P \cup N_S$.  The remaining
columns, the {\em semi-clique columns}, contain right endpoints from $P$ and left endpoints
from from $N \setminus N_S$, or vice versa.  The {\em non-simplicial columns} are
the clique and semi-clique columns.

At the highest level, the strategy is
to build up increasingly larger matrices that are submatrices of normal models
of increasingly larger subgraphs of $G$.   At each step, we permute
the columns from the matrix produced by the previous step so that it is a submatrix of a normal model
of the next larger subgraph, then add rows and columns to obtain
a normal model of the larger subgraph.  We halt when we either discover
along the way that $G$ is not a probe interval graph, or else when we return
a normal model of $G$.

The reader may find it helpful to refer
to the following sequence of matrices when reading the paper.  The sequence is an outline
of the steps of the algorithm.  

\begin{itemize}
\item $M_K$:  This is a consecutive-ones ordered clique matrix of $G[P]$, which is
an interval graph; we get this using the interval-graph recognition algorithm
of~\cite{BL76}.  If $M$ is a normal model of $G - N_S$, some consecutive-ones ordering
of $M_K$ is a submatrix of $M[P]$.  
In general, however, not every consecutive-ones ordering of $M_K$ is a submatrix of a normal model
of $G - N_S$.

\item $M^+_K$:  The columns of $M^+_K$ are the clique columns of every normal model of $G - N_S$,
which is the same set of columns as that of $M_K$.
We obtain $M^+_K$ by adding one row for each member of $N \setminus N_S$
to $M_K$, filling in the remainders of the columns in the new rows,
and finding a consecutive-ones ordering of the resulting matrix.
Not all consecutive-ones orderings of $M^+_K$
are submatrices of normal models of $G - N_S$, however.

\item $M'_K$ and $M^*_K$:  $M'_K$ is obtained by adding {\em constraint rows} to $M^+_K$. 
These
correspond to columns that must be consecutive in any ordering of columns of $M^+_K$
that gives a submatrix of a normal model of $G - N_S$.  $M^*_K = M'_K[V \setminus N_S]$
is a submatrix of a normal model of $G - N_S$.
The matrices $M^+_K$ and $M^*_K$ differ only in the order of their columns.
We let $C_K$ denote their column set.

\item $M_N$:  This is a normal model of $G - N_S$ obtained by adding semi-clique columns to $M^*_K$.
This is possible to do because $M^*_K$ is a submatrix of a normal model of $G - N_S$.  
We show that, ignoring members of $N_S$, this gives the clique and semi-clique columns in every
normal model of $G$.  Not all consecutive-ones orderings of columns of $M_N$ are submatrices
of normal models of $G$, however.

\item $M_P$:  For some normal model $M'_G$ of $G$, some ordering of columns of $M_P$ is
equal to the entire set of columns of $M'_G[P]$.  There is one column containing the probe set
of each clique column and each semi-clique column.  In addition, there is one column
for each neighborhood of a simplicial non-probe that is not the probe set of a clique
column or a semi-clique column.  The latter set corresponds to the simplicial columns
of a normal model of $G$.  Thus, for each $x \in N_S$, there is a column
of $M_P$ equal to $N(x)$. The set of rows of $M_P$ corresponds to the set of probes $P$.

\item $M_G[V \setminus N_S]$ and $M_G$:
Let $X$ denote the clique and semi-clique columns in a normal model $M_G$ of $G$.
$M_G[V \setminus N_S][X]$ is a consecutive-ones ordering of
the columns of $M_N$.  
The set columns of $M_G[V \setminus N_S]$ is the same set of columns as of $M_P$,
and the set of rows of $M_G[V \setminus N_S]$ is the same set of rows as of $M_N$.
Note that the rows of $M_P$ are a subset of the rows
of $M_N$ and the columns of $M_N$ are a subset of the columns of $M_P$.
We can find $M_G[V \setminus N_S]$ for some model $M_G$ of $G$
by solving the consecutive-ones probe matrix problem using $M_P$ in the role of $M_R$
and $M_N$ in the role of $M_C$, where $M_R$ and $M_C$ are the matrices from the definition
of the consecutive-ones probe matrix problem (Section~\ref{sect:classes}). Since we need an $O(n+m)$ time bound, and
not a time bound proportional to the number of 1's in $M_R$ and $M_C$, we
can explicitly fill in the $\ast$'s that are 1's, giving a sparse representation
of $M_G[V \setminus N_S]$.  This is because $M_G[V \setminus N_S]$ is a submatrix
of a normal model of $G$, so it has $O(n+m)$ 1's.

The probe set in each simplicial column is a subset of a clique of $G[P]$.  It follows
that $M_G[V \setminus N_S]$ is a model of $G - N_S$.  It is not a normal model of $G - N_S$.
However, for each $x \in N_S$, it now has a column equal to $N(x)$, so $x$ can be placed
in this column and its neighborhood is correctly represented.  Doing this for all
$x \in N_S$ yields a normal model $M_G$ of $G$.  
\end{itemize}

Let us assume that the members of $P$ occupy the top rows of
a model, followed by members of $N \setminus N_S$, followed
by the members of $N_S$.  This is accomplished with a suitable
numbering of the vertices, which fixes the row order in all models.

For each probe $p$, let ${\cal Q}(p)$ denote the set of cliques
of $G[P]$ that are subsets of $N[p]$.  These are the set of
cliques that contain $p$ as a member.  Generalizing this,
for each non-probe $x$, we let ${\cal Q}(x)$ denote the set of 
cliques of $G[P]$ that are subsets of $N(x)$.  Since the sets represented by columns
of $M_K$ are the cliques of $G[P]$, we may represent ${\cal Q}(p)$
and ${\cal Q}(x)$ with the corresponding sets
$Q(p)$ and $Q(x)$ of columns of $M_K$.
The set $Q(p)$ is given for each probe $p$, and
using the fact that $M_K$ is consecutive-ones ordered, we can efficiently
find $Q(x)$ for each non-probe $x$, as described in Section~\ref{sect:n1n2n3}.

We partition the set 
of non-probes into three sets: for a non-probe $x$, $x \in N_1$ 
if $|Q(x)| > 1$, $x \in N_2$  if it is non-simplicial and $Q(x) = \emptyset$,
and $x \in N_S$ if it is simplicial.
Note that if $x \in N_1$, it is not simplicial, so
$\{N_1,N_2,N_S\}$ is a partition of $N$.

To find $M^+_K$, 
we add one row to $M_K$ for each $x \in N_1 \cup N_2$, equal
to $Q(x)$.  That this completes the clique columns
of every normal model of $G - N_S$ follows
from the Helly property and the appearance of each clique of $G[P]$ in 
a unique column.  
For each vertex of $N_2$, the row is empty.
$M^+_K$ is any consecutive-ones ordering of the resulting
matrix.

For each $x \in N_1$,
probes in columns in $Q(x)$ are neighbors of $x$, since they
belong to cliques that are subsets of $x$'s neighborhood.
$M^+_K$, when interpreted as a probe interval model, represents these adjacencies.
$M^+_K$ is not a complete model of $G - N_S$; for a neighbor $p$ of $x$, it can be that $x$
only occurs in columns that also have non-neighbors of $p$.   Therefore, $Q(p)$ and $Q(x)$ do 
not intersect in $M^+_K$, and $M^+_K$ fails to represent their adjacency.
Let us call these {\em unfulfilled adjacencies} in $M^+_K$.

These are the reason semi-clique columns also occur in normal models of $G - N_S$.
Such an unfulfilled adjacency must be resolved by inserting a semi-clique
column where $x$ and $p$ can meet.
This requirement places additional constraints on the ordering
of columns of $M^+_K$.  It must be such that no columns $c$ 
intrudes between the intervals for $x$ and $p$.  This would
block them from meeting each other in a new column, since one of them would
have to cross $c$.  Since  $c$
is not a member of $Q(x)$ or of $Q(p)$, it
contains non-neighboring probes for both $x$ and $p$, and having one
of them cross $c$ would misrepresent the graph.

We can avoid this by adding a {\em constraint row} to $M^+_K$,
equal to $Q(x) \cup Q(p)$ for each such pair, and getting a consecutive-ones
ordering of the resulting matrix.  (See Figure~\ref{fig:constraints}).  
A variant of this trick is required when $x \in N_2$, since then $Q(x)$ is
empty.  

Doing this for all unfulfilled adjacencies would exceed the $O(n+m)$ bound.
Fortunately, many of the constraints are redundant.  We add a constraint only for
{\em representative pairs} of unfulfilled adjacencies, and the omitted
constraints are redundant (Figure~\ref{fig:minConstraints}).  This adds
$O(n+m)$ 1's to the matrix, and a consecutive-ones ordering of it
gives $M'_K$.

Now that they have helped order the columns of $M'_K$, we delete the constraint
rows to obtain $M^*_K = M'_K[V \setminus N_S]$, which
is a submatrix of a normal model of $G - N_S$.

We can now extend $M^*_K$ to the model $M_N$ of $G - N_S$
by inserting the semi-clique columns between each pair of consecutive
columns of $M^*_K$ (Figure~\ref{fig:Cxy}).

A normal model can then be built for $G$ using a solution to the
consecutive-ones probe matrix problem, as described in the last point above.
A key tool in our solution of the consecutive-ones probe matrix
is the idea of the {\em restriction} $T[C]$ of the PQ tree $T$ of
a matrix $M$ to the one-element subsets of a set $C$ of its columns.  This is a type of homomorphism 
that preserves the constraints consecutive-ones orderings of $M$ impose
on the relative orderings of columns in $C$.  (See Figure~\ref{fig:PQRestrict}.)
A special case of this operation has been described in~\cite{McCC1Cert} and in~\cite{McCPQAlg}, but
only for certain sets of columns such that $T[C]$ is the PQ tree of $M[][C]$.
The general solution is trivial to compute (Algorithm~\ref{alg:PQRestrict}).
The concept shows promise as a tool for proofs.
The strategy for applying it to the consecutive-ones probe matrix problem
is outlined in the caption of Figure~\ref{fig:probeOnes}.

\section{Initial steps and observations}\label{sect:initial}

We can run the recognition algorithm separately on each connected component
of $G$ to produce disjoint probe interval models for the components.
The collection of these is a probe interval model of $G$.
If any component fails to be a probe interval graph, then $G$ fails
to be a probe interval graph.  This reduces the problem to that of deciding
whether a connected graph is a probe interval graph, and producing a model if it
is. Henceforth, we will
assume that $G$ is connected.

\subsection{Variations on radix sorting}

Given a collection of lists of integers from $\{1, 2, \ldots, n\}$, whose
sum of lengths is $k$, we 
may sort each list by sorting all the elements of all the lists in a single radix sort
using set number as primary sort key and element value as secondary
sort key.  This takes $O(n+k)$ time.
We can sort the adjacency lists of a graph in linear time, for example.
Also, we can sort the collection of the lists lexicographically in $O(n+k)$ time
even though they have different lengths~\cite{AHU74}.  Thus, we can sort
the rows or columns of a matrix lexicographically in linear time, given
a sparse representation.

Variations of this that we will use are the following.
Suppose we are given $j$ groups of lists of integers from 1 through $n$,
and the sum of lengths of the lists of integers is $k$.
We can sort each list of integers, all lists lexicographically,
and then, using a stable sort, segregate this list back into
the $j$ groups.  This gives each of the $j$ groups, sorted lexicographically,
in $O(n+k)$ time.  Since each list of integers is sorted,
we may eliminate any
duplicate lists in any of the $j$ groups, also in $O(n+k)$ time.
If, instead of sorting
the lists lexicographically in this sequence of operations, we sort 
the lists by length, we get each of the $j$ groups
sorted by length.  We can then determine whether the elements of each of the $j$ groups
induce a chain $X_1 \subseteq X_2 \subseteq \ldots \subseteq X_k$  
in the subset relation, in $O(n+k)$ time.  

\subsection{Properties of normal models}\label{sect:normalModels}

\begin{lemma}\label{lem:normalModel}
For every probe interval graph $G=(P,N,E)$, there exists a normal probe 
interval model of $G$.
\end{lemma}
\begin{proof}
Let $M$ be an arbitrary probe interval model of $G$.
If some row is not taut, we may change an endpoint of its
interval from 1 to 0 without affecting the represented graph.  
If two consecutive columns can be merged without changing the represented 
graph, we merge them.

We iteratively perform one of these operations until none of them can be 
performed. Since each operation reduces the number of 1's in the matrix,
this process eventually results in a normal model,
and since none of the operation changes
the represented probe interval graph, it is a normal model of $G$.
\end{proof}

\begin{lemma}\label{lem:nonemptyProbes}
In every normal model of a connected probe interval graph that has more than one vertex,
every column has a nonempty set of probes.
\end{lemma}
\begin{proof}
If $c$ is a column with an empty set of probes, $c$ cannot be
the endpoint of any vertex, since it would either fail to be
taut or be a simplicial non-probe with no neighbors, contradicting
connectedness of the graph.
\end{proof}

\begin{lemma}\label{lem:simplicialPoints}
In every normal model, each non-probe resides
in a single column if and only if it is simplicial.
\end{lemma}
\begin{proof}
If a non-probe $x$ occurs in only one column $c$, its neighbors are $S(c) \cap P$, which form
a complete subgraph due to their presence in a common column, and $x$ is simplicial.
If a simplicial non-probe $y$ has a 1 in more than one column, 
the Helly property dictates that one of the columns has
$N(x)$ as its probe set, and $y$'s 1's in the other columns
can be deleted without affecting the represented graph,
contradicting the tautness of the model.  
\end{proof}

A probe may occur in only one column without being simplicial, since its 
non-probe neighbors are nonadjacent to each other.

\begin{lemma}\label{lem:cliquesUnique}
In a normal model, each clique of $G[P]$ is a subset of exactly one column.
\end{lemma}
\begin{proof}
Each clique of $G[P]$ must be a subset of at least one column by the Helly property of interval models.
Suppose a clique $X$ is a subset of more than one column.  By consecutiveness of 1's in each row of $X$,
the intersection of these rows is consecutive.  The columns in this intersection
all have $X$ as their probe set.  
They can be merged, since the intervals that now meet and didn't meet before
are non-probes, contradicting the normality of the model.
\end{proof}

We can now fill out and justify the classification of columns as clique, semi-clique, and
simplicial columns.

\begin{definition}
A vertex that is a member of only one column $c$ in a model has {\em degenerate endpoints}; they 
are both $c$.  The endpoints of a vertex that is a member of more than one column
are {\em proper endpoints}.
A column $c'$ 
of a model is a {\em clique column} if it contains a clique of $G[P]$.
It is a {\em left semi-clique column} of $S(C) \cap (V \setminus N_S)$ if it has a proper right
endpoint of a probe and a proper left endpoint of a non-probe, 
no left endpoint of a probe, and no right endpoint of
a non-probe.  A {\em right semi-clique column} is defined symmetrically.
It is a {\em simplicial column} if it is not a clique column, contains a simplicial
non-probe, and contains no endpoints of non-simplicial non-probes.
\end{definition}

\begin{lemma}\label{lem:C1C2C3}
In a normal model, every column is a clique column, a semi-clique column,
or a simplicial column.
\end{lemma}
\begin{proof}
Let $c$ be a column and assume for contraction that it is not of one of these three types.  

Suppose that $c$ is not the endpoint of any probe. If $c$ is a proper endpoint of some
non probe $x$ then $x$ is not taut.
If $c$ is a degenerate endpoint of a non-probe then it is a simplicial column.
If $c$ is not an endpoint of any non-probe, then it can be merged with one
of the adjacent columns without affecting the represented graph.

Suppose that $c$ is an endpoint of a probe $p$.
Without loss of generality, suppose
it is a right endpoint.  The column $c$ does not contains a left endpoint of a probe, 
since otherwise it would be a clique column.
Since $x$ has no left endpoint of a probe, no member of
$N_1 \cup N_2$ can have its right endpoint in the column, as it would not be
taut.  If $c$ has no left endpoint at all in the column, then $p$ is not taut.
Thus $c$ must have the left endpoint of a non-probe.
If it contains the left endpoint of a member of $N_1 \cup N_2$, then
it satisfies the definition of a left semi-clique.  Otherwise, the left endpoints
in the column belong to simplicial non-probes, and it is a simplicial column.
In any of the cases we either get a contradiction to the normality of the model,
or to the assumption that $c$ is not of one of the three types. Therefore,
every column must be of one of the tree types.
\end{proof}

Though it would be convenient, we
cannot require that every column have the endpoint of a probe.
When a simplicial non-probe has as its neighbors the
intersection of consecutive cliques of $G[P]$ in the model, its
column cannot contain the endpoint of a probe in the model.

\begin{lemma}\label{lem:noSubsets}
In a normal model, no column is a subset of any other.
\end{lemma}
\begin{proof}
Suppose a column $c$ is a subset of a column $c'$.  Without loss
of generality, suppose $c'$ is to the left of $c$.  By consecutiveness of 1's,
if $c''$ is the adjacent column to the left of $c$
$S(c) \subseteq S(c'')$.  Thus, $c$ can be merged with $c''$ without changing
the represented graph.  Since they are consecutive, this does not affect
consecutiveness of 1's in the model.
\end{proof}

We call a column $c$ that is not a clique column a {\em non-clique column}.

\begin{lemma}\label{lem:nextProbes}
In a normal model $M$, the probe set in every
non-clique column is a proper subset of either the probe set
of the next column to its left or of the probe set of the next column to its right.
\end{lemma}
\begin{proof}
If $c_1$ and $c_2$ are consecutive columns with equal probe sets,
they can be merged without affecting consecutiveness of 1's or
the represented graph.  No consecutive columns have the same probe set.

Let $c$ be a non-clique column.  A non-clique column cannot contain both left and right endpoints
of probes.  Without loss of generality, suppose that it does not contain
the right endpoint of a probe.  Then the set of probes in $c$ is a subset of
the set of probes in the column to its left.  Since consecutive columns cannot
have equal probe sets, the probes in $c$ must be a proper subset of the probes
in the column to its left.
\end{proof}

As we mentioned above, the matrix clique of an interval graph has $O(n)$ columns
and $O(m + n)$ 1's. It is not obvious that a model of a probe interval graph maintains this
property, since there might be $\Theta(n^2)$ adjacencies among non-probes which are
realized by the model but are not represented by edges in the graph. The next lemma
shows however that for normal models this property holds.

\begin{lemma}~\label{lem:boundedOnes}
If $M$ is a normal model of a probe interval graph $G$,
then $M$ has at most $n$ columns and $O(n+m)$ 1's.
\end{lemma}
\begin{proof}
By Lemma~\ref{lem:C1C2C3}, every column contains
a left endpoint and a right endpoint.  There are $2n$ endpoints,
so the number of columns is at most $n$.

If a column has the endpoint of a probe $p$, charge the 1's
in the column to $p$.
The number of 1's in the column is 
bounded by the size of the closed neighborhood of $p$.  Over
all columns, each probe is
charged in at most two columns for the size of its closed
neighborhood, so the number of 1's in these columns is $O(n+m)$.

It remains to bound the number of 1's in simplicial columns.
Let $\{c_1, c_2, \ldots, c_k\}$ be the simplicial columns.
For each $j$ from 1 through $k$, let $p_j$ be a probe with an endpoint
in $c_j$, or if it has no endpoint of a probe, let $p_j$
be a probe with an endpoint in a column next to $c_j$.
By Lemma~\ref{lem:nextProbes}, $p_j$ exists.
Since $c_j$ contains no endpoint of a vertex in
$N_1 \cup N_2$, by definition,
every member of $V \setminus N_S$ in $c_j$
is a neighbor of $p_j$.  The number of 1's in $c_j[V \setminus N_S]$ 
is at most $|N[p_j]|$.  Charge these 1's to $p_j$.
Charge the 1's in $c_j[N_S]$
to an arbitrary probe $q$ in $c_j$; the rows where these 1's
occur are all neighbors of $q$.

Each probe is charged $O(1)$ times in the role of $p_j$, for 
$|N[p_j]|$ 1's.
A probe could be charged many times in the role of $q$, but never
twice for the same neighbor, since they are all simplicial and occur in 
only one column, by Lemma~\ref{lem:simplicialPoints}. The total number
of these charges to $q$ is bounded by $|N(q)|$.
Summing these charges over all probes gives the $O(n+m)$ bound on
the number of 1's in columns $\{c_1, \ldots, c_k\}$.
\end{proof}

To derive more properties of normal models, we make use
of the following insight, which is due to Zhang~\cite{Zhang94} (see also~\cite{McMorris98}).
Let $E'$ be the edges of $G -  N_S$.
He defined the set $E^+ = \{xy \mid x, y \in N_1 \cup N_2$
and $N(x) \cap N(y)$ contains two nonadjacent
vertices.$\}$.  He then showed that an interval model
of $G^* = (V \setminus N_S, E' \cup E^+)$ is a probe interval model
of $G -  N_S$.  

The strategy of a step of the probe interval graph recognition algorithm of Uehara~\cite{U04}
is to construct $G^*$ in order to
find a model of $G - N_S$.  We cannot use that approach
because $G^*$ does not have $O(n+m)$ edges.  A simple example of this
is a graph with two nonadjacent vertices,
$p_1$ and $p_2$, and $n-2$ non-probes, each adjacent to $p_1$ and $p_2$.
The probe interval graph has $O(n)$ edges, but in
$G^*$, the $n-2$ non-probes form a complete subgraph, so $G^*$
has $\Theta(n^2)$ edges.

However, we can derive structural properties of normal models
from it by observing a normal model of $G -  N_S$ is an interval
model of a slight variation of $G^*$.  Therefore, even though this
graph does not have $O(n+m)$ edges, every clique matrix has $O(n+m)$ 1's.

\begin{definition}
Let $E'$ be the edges of $G - N_S$.
Let $E^{++} = \{xy \mid x,y \in N_1 \cup N_2$, and $N(x) \cap N(y)$
is either a clique of $G[P]$ or contains two nonadjacent
vertices$\}$.
Let $G^{**} = (V \setminus N_S, E' \cup E^{++})$.
\end{definition}

\begin{lemma}\label{lem:UInverse}
Every normal model $M$ of $G - N_S$ is a consecutive-ones 
ordered clique matrix of $G^{**}$.
\end{lemma}
\begin{proof}
Let $M$ be a normal model of $G - N_S$.
Let $x$ and $y$ be members of  $N_1 \cup N_2$.

If $N(x) \cap N(y)$ is a clique $K$ of $G[P]$, their intervals
must intersect at the only column that contains $K$, by Lemma~\ref{lem:cliquesUnique}.

If $N(x) \cap N(y)$ contains two nonadjacent
vertices, $p_1$ and $p_2$, suppose without loss of generality
that $p_1$ lies to the left of $p_2$ in $M$.  Let $c$
be the rightmost column of $p_1$ and $C'$ be the leftmost
column of $p_2$.
Since they are both adjacent to $p_1$ and $p_2$, 
$x$ and $y$ must be contained in both $c$ and $c'$.

If $N(x) \cap N(y)$ is empty, then $x$ and $y$ contain
no column in common, by Lemma~\ref{lem:nonemptyProbes}.

Otherwise, $N(x) \cap N(y)$ is a complete subgraph that is not
a clique of $G[P]$.  Then $N(x) \cap N(y)$ is a proper subset
of a clique $K$.  The intersection of $x$ and $y$
is a consecutive set $Y$ of columns that do not contain $K$.
Suppose without loss of generality that the clique column
containing $K$ lies to the left of $Y$, and that the
right endpoint of $x$ is at the right endpoint of $Y$.  
By Lemma~\ref{lem:simplicialPoints}, $x$ 
is contained in more than one column of $M$.  
The right endpoint of $x$ is not taut, contradicting the
normality of the models, so this case cannot happen.

It follows that $x$ and $y$ are contained in a common column
if and only if $xy \in E^{++}$.  For any pair $\{u,v\}$
where at least one of $u$ and $v$ is a probe, $u$ and $v$
are contained in a common column if and only if
$uv \in E'$, by the definition of a probe interval model.
$M$ is an interval model of $G^{**}$, and so it is
a consecutive-ones ordered clique matrix of $G^{**}$ as
required.
\end{proof}

\begin{lemma}\label{lem:U}
When interpreted as a probe interval model, every
consecutive-ones ordered clique matrix of $G^{**}$ is 
a normal model of $G - N_S$
\end{lemma}

\begin{proof}
Let $M$ be a consecutive-ones ordering of a clique matrix of $G^{**}$
that is a normal model of $G - N_S$.  $M$ exists by Lemma~\ref{lem:normalModel}
and Lemma~\ref{lem:UInverse}.
Let $M'$ be a different consecutive-ones ordering of $M$.
Suppose that $M'$ is not a normal
model of $G - N_S$.  Then it can be turned into a normal
model $M''$ by a series of contractions of endpoints
and merges of columns, as in the proof of Lemma~\ref{lem:normalModel}.
Each of these operations reduces the number of 1's in the matrix.
Therefore $M''$ has fewer 1's than $M'$, hence
fewer 1's than $M$.  Since $M''$ is a normal model of 
$G - N_S$, it is a
consecutive-ones ordering of $M$ by Lemma~\ref{lem:UInverse}, but
$M''$ and $M$ have different numbers of 1's, a contradiction.
\end{proof}

The following is immediate from Lemma~\ref{lem:UInverse}
and~\ref{lem:U}.

\begin{theorem}\label{thm:allModels}
The normal models of $G - N_S$ are the consecutive-ones
orderings of a single matrix.
\end{theorem}

It follows that the set of all normal models of $G - N_S$
is given by the PQ tree of any one of these models.

According to the following theorem, ignoring members in $N_S$, the collection of clique columns and
the collection of semi-clique columns are each invariant 
over all normal models of $G$, and no two clique or semi-clique columns are
equal.

\begin{theorem}\label{thm:invariantCols}
Let $M$ and $M'$ be normal models of $G$.  No two clique or semi-clique columns of 
$M[V \setminus N_S]$ are equal.  Let ${\cal C}$ and ${\cal C'}$ be the collection of sets of vertices 
represented by clique and semi-clique columns of $M[V \setminus N_S]$ and $M'[V \setminus N_S]$, 
respectively.  Then ${\cal C}$ = ${\cal C'}$.
\end{theorem}
\begin{proof}
In a normal model $M_G$ of $G$, let $X$ be the set of non-simplicial
columns of $M_G$.  Let $M'_N = M_G[V \setminus N_S][X]$.
Two probes that are contained in a simplicial column $c$ of $M_G$ are contained in a neighboring
column $c'$, by Lemma~\ref{lem:nextProbes}.  They remain adjacent
when $c$ is deleted.  Since a simplicial column
contains no endpoints of non-probes of $V \setminus N_S$, non-probes that are contained
in $c$ are contained in $c'$.  They remain adjacent to the probes
in $c$ when $c$ is deleted.
Deleting a simplicial columns from $M_G$ does not change the graph
represented by $M_G[V \setminus N_S]$.
Therefore, $M'_N$ is a model of $G - N_S$.  

If $d$ is a clique column of $M_G$, then it contains left and right endpoints
of probes, and this remains true in $M'_N$.  Every endpoint
in the column is taut.  

If $d$ is a semi-clique
column, assume without loss of generality that it contains left
proper endpoints of non-probes in $M_G$.  These non-probes are
members of $N_1 \cup N_2$ by Lemma~\ref{lem:simplicialPoints},
and $d$ contains left proper endpoints of non-probes in $M'_N$.
It contains no right proper endpoints of non-probes in $M_G$ by
the properties of semi-clique columns.
Since no simplicial column contains a proper endpoint of a non-probe,
$d$ contains no right endpoints of non-probes in $M'_N$.
It contains right endpoints of probes in $M_G$, and this
remains true in $M'_N$.  Every endpoint in $c$ is taut in
$M'_N$.  

Suppose two columns $c_1$ and $c_2$ of $M'_N$ can be merged.  
Then they are consecutive and have the same probe set $Y$. 
Then since no clique column was deleted from $M_G$ to
obtain $M'_N$,  $Y$ is the probe set in every column 
between $c_1$ and $c_2$ in $M_G$, and all columns in
this interval can be merged in $M_G$, contradicting
the normality of $M_G$.  

$M'_N$ is taut and minimal, so it is a normal model.
No two columns of $M'_N$ are equal, by Lemma~\ref{lem:noSubsets}.
The lemma now follows from Theorem~\ref{thm:allModels}.
\end{proof}

Unfortunately, it is not the case that every 
consecutive-ones ordering of a normal model of $G$ is a normal model if $G$, due to the
presence of simplicial non-probes.  Figure~\ref{fig:naughtySimplicials} gives an example.
There does not seem to be a single PQ tree for representing the possible models
of a probe interval graph once simplicial non-probes are introduced.  

\begin{figure}
\centerline{\includegraphics[]{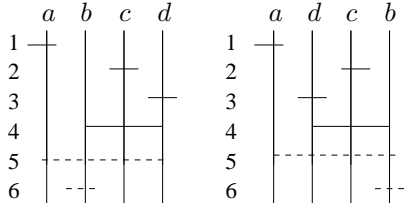}}
\caption{Not every consecutive-ones ordering of a normal model is necessarily a normal model
if it contains simplicial non-probes.
The example on the left is a normal model, with probes $\{1,2,3,4\}$, non-probes
$\{5,6\}$, and simplicial non-probe 6.  The one on the right is a consecutive-ones
ordering of it, but the right endpoint of interval for vertex 5 is not taut.  There is 
a PQ tree for representing all normal models of a probe interval graph if it has
no simplicial non-probes, but there is not a simple PQ tree for representing the possible 
arrangements of intervals of a probe interval graph once simplicial non-probes are introduced.
Any representation of all models, or those models of some class such as normal models,
must be at least as expressive as the gadget
from Figure~\ref{fig:probeOnes}, part I, and it gives families of permutations that
no single PQ tree can give.
}\label{fig:naughtySimplicials}
\end{figure}

\subsection{PQ trees}\label{sect:PQtrees}

We defined PQ-tree in Section~\ref{sect:intro}.
The classification of a node as a P node or a Q node is ambiguous
if it has only two children.  In this paper we adopt the convention of considering
it to be both.

A PQ tree can be represented in $O(n)$ space by letting each leaf carry a column identifier and each
internal node carry a pointer to an ordered list of its children.  Notationally, however, we denote each
node in a PQ tree by a set, namely, the set of columns at leaf descendants. A leaf is a
set whose only element is a column, and an internal node is the disjoint union of its children.  
The root is the set of all columns of the matrix.  A consecutive set of columns in one consecutive-ones
ordering must be consecutive in every consecutive-ones ordering if and only if is a P node or 
a union of consecutive children of a Q node.

\begin{definition}\label{def:validOrderings}
Let $T$ be a PQ tree and
$\pi()$ be  a bijection from its leaves
to $\{1,2, \ldots, k\}$, such that there is an allowed leaf order
where each leaf $\ell$ is in position $\pi(\ell)$ in the ordering.
Then $\pi()$ is a {\em valid ordering} for $T$.
Let $\Pi(T)$ the set of valid orderings for $T$.
\end{definition}

\begin{definition}\label{def:TreePrec}
If $T$ and $T'$ are two PQ trees with the same leaf set,
let $T \prec T'$ denote that $\Pi(T) \subset \Pi(T')$, 
$T \preceq T'$ denote that $\Pi(T) \subseteq \Pi(T')$,
and let $T \equiv T'$ denote that $\Pi(T) = \Pi(T')$.
\end{definition}

\subsubsection{Restricting a PQ tree}

\begin{definition}\label{def:perms}
If $\pi$ is a bijection from a set $C$ to $\{1,2, \ldots, |C|\}$ and $X$ is a nonempty subset of
$C$, then $\pi[X]$ denotes the bijection from $X$
to $\{1,2, \ldots, |X|\}$ such that for $a,b \in X$, $\pi(a) < \pi(b)$
if and only if $\pi[X](a) < \pi[X](b)$.  If $\Pi$ is a set of permutations
from $C$ to $\{1,2, \ldots, |C|\}$, then $\Pi[X]$ denotes $\{\pi[X] \mid \pi \in \Pi\}$.
\end{definition}

\begin{definition}\label{def:PQRestriction}
Let the {\em restriction $T[X]$ of a PQ tree $T$ to $X$ denote
the PQ tree $T'$ such that $\Pi(T') = \Pi(T)[X]$}.
If $T$ has leaf set $C$, $|C| > 1$, and $\{c\}$ is a leaf,
let $T - c$ denote $T[C - c]$.
\end{definition}

That $T[X]$ is well-defined can be seen from the following 
algorithm that computes it, in time linear in the number of nodes of $T$. 
(See Figure~\ref{fig:PQRestrict}.)

\begin{algorithm}\label{alg:PQRestrict}
Delete leaves that of $T$ are not in $X$ and each node that has no leaf descendants 
in $X$.  Then, for each node $u$ that has only one child $w$, we can replace $u$ with $w$
in the ordered list of children at $w$'s parent, since $u$
imposes no constraints on the orderings of $X$.
\end{algorithm}

\begin{figure}
\centerline{\includegraphics[]{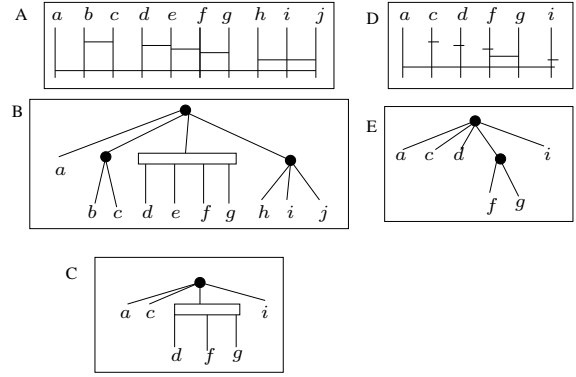}}
\caption{The restriction of a PQ tree to a subset of columns.
A.  A consecutive-ones ordered matrix $M$.  Rows of 1's are depicted with line segments.
B.  The PQ tree $T$ of $M$.  C.  The restriction $T[C]$ of $T$ to columns
$C = \{a,c,d,f,g,i\}$.  It is obtained by Algorithm~\ref{alg:PQRestrict}, and gives
all orderings of $\{a,c,d,f,g,i\}$ that are subsequences of consecutive-ones
orderings of $M$.  D.  The submatrix $M[][C]$ of $M$ given by columns in $C$.
E.  $T(M[][C])$, showing that this is not
the same as $T[C]$.  $T[C]$ retains constraints on the ordering that have been lost
from $M[][C]$.
}\label{fig:PQRestrict}
\end{figure}

\begin{lemma}\label{lem:RestrictPrec}
If $M$ is a consecutive-ones matrix and $C$ is a subset of its columns,
then $T(M)[C] \preceq T(M[][C])$.
\end{lemma}
\begin{proof}
Every submatrix of a consecutive-ones ordered matrix is consecutive-ones
ordered, so for every $\pi \in \Pi(T(M))$, $\pi[C] \in \Pi(T(M[][C]))$.
\end{proof}

The opposite direction of Lemma \ref{lem:RestrictPrec} does not always hold,
as shown by part E of Figure~\ref{fig:PQRestrict}.

\subsubsection{A relationship between the PQ tree and the rows of the matrix}

Let us say that two rows $X$ and $Y$ {\em properly overlap}
if $X \cap Y$, $X \setminus Y$, or $Y \setminus X$ are all
nonempty.  The following is easily verified, and has appeared
in~\cite{MPT98}, among other places:

\begin{lemma}\label{lem:overlap}
If $X$ and $Y$ are properly overlapping rows of a consecutive-ones matrix $M$
then $X \setminus Y$, $X \cap Y$, $Y \setminus X$ and $X \cup Y$ are
consecutive in every consecutive-ones ordering of $M$.  
\end{lemma}

\begin{definition}\label{def:CPerp}~\cite{MPT98}
If $M$ is a consecutive-ones matrix on column set $C$.
Let ${\cal C}^{\perp}(M)$ denote $\{X \mid \emptyset \subset X \subseteq C$
and $X$ does not properly overlap any row of $M\}$.
\end{definition}

\begin{lemma}\label{lem:CPerp}~\cite{MPT98}
If $M$ is a consecutive-ones matrix,
${\cal C}^{\perp}(M) = \{X \mid X$ is a node
of $T(M)$ or a nonempty union of children of a 
P node of $T(M)\}$.
\end{lemma}

\begin{definition}\label{def:FPerp}
Suppose $M$ is consecutive-ones ordered.  
Let ${\cal F}(M)$ denote those members of ${\cal C}^{\perp}(M)$
that are consecutive in $M$, that is, the nonempty consecutive sets of
columns that do not properly overlap any row of $M$.
\end{definition}

Conceptually, the members of ${\cal F}(M)$ are those consecutive
sets of columns whose order can be reversed to give a new consecutive-ones
ordering of $M$.  This gives the following:

\begin{lemma}\label{lem:FPerp}
Let $M$ be a consecutive-ones ordering of a 0-1 matrix, and let 
$T$ be the corresponding ordering of its PQ tree.
Then ${\cal F}(M) = \{X \mid X$ is a node of $T$ or a nonempty union
of consecutive children of a Q node of $T\}$.
\end{lemma}

$T(M)$ and ${\cal F}(M)$ are equivalent representations of the constraints
on the consecutive-ones orderings of columns of $M$, but it is sometimes easier
to prove properties of $T(M)$ by expressing them in terms of ${\cal F}(M)$.

\subsubsection{Finding intersections of PQ trees}

In~\cite{McCPQAlg}, it is shown that if $T$ and $T'$
are PQ trees with the same leaf set and $\Pi(T) \cap \Pi(T')$
is nonempty, then $\Pi(T) \cap \Pi(T')$ can also be represented
with a PQ tree, denoted $T \cap T'$.  This is easy to see:
if $M$ is a matrix with $T$ as its PQ tree and $M'$
is a matrix with $T'$ as its PQ tree, then the matrix
$M''$ whose rows are the union of rows of $M$ and $M'$ gives a matrix 
whose PQ tree is represents $\Pi(T) \cap \Pi(T')$, unless this
set is empty, which is the case if and only if $M''$ does not
have the consecutive-ones property.

\begin{definition}\label{def:PQintersection}
Let $T$ and $T'$ be PQ trees with the same leaf sets.  By $T \cap T'$,
we denote the PQ tree $T''$ such that $\Pi(T'') = \Pi(T) \cap \Pi(T')$,
unless this set is empty, in which case we say the intersection
of $T$ and $T'$ is {\em undefined}.
\end{definition}

Booth and Lueker showed that every consecutive-ones matrix has
a PQ tree.  We observe that the converse also applies, by the
following construction, which, given a PQ tree $T$, constructs
a canonical matrix $M(T)$ that has $T$ as its PQ tree.

\begin{itemize}
\item For each P node $p$ that is not the root, let $p$
be a row of $M(T)$;

\item For each Q node $q$, let $(C_1, C_2, \ldots, C_k)$ be the left-to-right
order of the children.  For
every consecutive pair $(C_i, C_{i+1})$ of children,
let $C_i \cup C_{i+1}$ be a row of $M(T)$.
\end{itemize}

The correctness of this is immediate from Lemma~\ref{lem:FPerp}.

\begin{lemma}\label{lem:MofT}
If $T$ is the PQ tree of a matrix $M$ with $n$ columns and rows and $m$ 1's, 
$M(T)$ has $O(n+m)$ 1's and a sparse representation takes 
$O(n+m)$ time to generate.
\end{lemma}
\begin{proof}
Let $T$ be the PQ tree of a matrix $M$ with $n$ columns and rows and $m$ 1's.
In various sources, for example,
in~\cite{MPT98}, it is shown that
if $A$ is a Q node and $C_i$ and $C_{i+1}$
are consecutive children, then for some row $x$ of $M$,
$A$ is the least common ancestor of the members of $S(x)$, and
$C_i \cup C_{i+1} \subseteq S(x)$.  
Since $C_i \cup C_{i+1}$ is a row of $M(T)$, we may charge
the 1's in this row to the 1's in columns of $C_i \cup C_{i+1}$ in $S(x)$.
The 1's of $S(x)$ in columns of $C_i$ are charged at most twice, at most once when $C_i$ appears in
$C_i \cup C_{i+1}$ and
at most once when $C_i$ appears in $C_{i-1} \cup C_i$.
Over all Q nodes, this charges each 1 in $M$ for at most two 1's in $M(T)$.

It is also shown in~\cite{MPT98} that if $B$ is a P node and not the 
root, then either there exists a row $w$ of $M$ such that $B = S(w)$,
or there exists a row $y$ such that the least common ancestor of $S(y)$ is a Q node parent $A$, and 
$B \subset S(y)$.  Charge the 1's
of the row corresponding to $B$ in $M(T)$ to the 1's in columns of $B$ in either $S(w)$ or $S(y)$, 
whichever exists.
Over all P nodes, this charges each 1 in $M$ for at most one 1 in
$M(T)$.  If the root is a P node, charge the 1's in the corresponding row of $M(T)$
to columns of the matrix.  Thus, the number of 1's in $M(T)$ is at most the number
of columns of the matrix plus three
times the number of 1's in $M$, which is linear in the size of $M$.
\end{proof}

An $O(n)$ algorithm is
given in~\cite{McCPQAlg} for finding this tree, but it uses sophisticated techniques and a roundabout
set of reductions in order to get this time bound.  
Since we do not need this bound, we use the following
straightforward method in $O(m + n)$ time:  

\begin{algorithm}\label{alg:PQIntersect}
Let $M$ be a matrix whose
rows are the union of rows in $M(T)$ and $M(T')$,
and use Booth and Lueker's algorithm to either generate
the PQ tree of $M$, which is $T \cap T'$, or determine 
that $M$ has no consecutive-ones ordering, in which
case $T \cap T'$ does not exist.
\end{algorithm}

\subsection{Finding $N_1$, $N_2$, and $N_S$ and $Q(x)$ for each $x \in N$}
\label{sect:n1n2n3}

To partition $N$ into $N_1$, $N_2$, and $N_S$,
we begin by finding an arbitrary consecutive-ones ordered clique matrix $M_K$ of 
$G[P]$. That is, we find a normal model of $G[P]$.  If none exists, we reject $G$,
since a requirement for $G$ to be a probe interval graph is for $G[P]$ to
be an interval graph.

We temporarily number the columns of $M_K$ from left to right.
Let $x \in N$ be a non-probe.   We find for every $p \in N(x)$
the left endpoint and the right endpoint of $p$. We keep the
column numbers of these two endpoints, together with their side (left or
right) in a list $L_x$. We 
radix sort the concatenation of these lists 
with $x$ as the primary sort key, column number
as the secondary sort key, and left versus right endpoint as
the tertiary key.  This gives each list $L_x$ in sorted order,
with left endpoints in a column preceding right endpoints in
the same column.  The time required is proportional to the
sum of cardinalities of these lists, $O(m)$.

We sweep through $L_x$ from left to right, keeping a running count of the number of neighbors 
of $x$ in the current column. Each time we 
encounter a left endpoint in $L_x$ we increment the counter, and each 
time we encounter a right endpoint we decrement it.  Each time we 
encounter a right endpoint $e$ that follows a left endpoint, we 
compare the current value of the counter with the size of the clique $K$ 
represented by the
column of $e$, and include $K$ in $Q(x)$ if they are equal.

To find out whether $x$ is simplicial, we test whether the counter reached the size of $N(x)$
at some point.  If it passes this test, $x$ is a member of $N_S$.
If $Q(x)$ is empty but $x$ is not simplicial, it is a member of $N_2$.
Otherwise, it is a member of $N_1$.

These procedures for $x$ take time proportional to $|N[x]|$ for every non-probe $x$.
Summing over all $x$, we have an $O(n+m)$ bound for these operations.
Summarizing, we get the following.

\begin{lemma}\label{lem:n1n2n3}
In linear time we can either split $N$ into $N_1, N_2$ and $N_S$ and 
find $Q(x)$ for every $x \in N$, or else determine that $G$ is not a 
probe interval graph.
\end{lemma}

\section{Finding a normal model $M_N$ of $G - N_S$}

\subsection{Finding $M^+_K$}

Recall that $M^+_K$ is the matrix $M_K$ with additional row for each non-simplicial non-probe.

Suppose $x \in N_1 \cup N_2$.  If clique $j$ of $G[P]$ is a member of $Q(x)$,
then $x$'s row must have a  1 in column $j$.  This follows
from the Helly property 
and the fact that there is only one column in a normal model that contains
clique $j$.  If clique $j$ is not in $Q(x)$,
then $x$'s row cannot have a 1 in $j$'s column, since this
would falsely represent $x$ as a neighbor of all members
of clique $j$.  
If $G - N_S$ is a probe interval graph,
this matrix therefore gives us the clique columns of every
normal model of $G - N_S$.
Note that the new rows for $N_2$ are empty sets; 1's will be added to them
later when new columns are added.

Let $M$ be the ordering of these columns in some normal model of
$G - N_S$.
For each $v \in P \cup N_1$, $Q(v)$ must be consecutive 
in $M$, since a submatrix of a consecutive-ones ordered matrix is consecutive-ones
ordered. We find a consecutive-ones ordering $M^+_K$ of the columns, and if no such
an ordering exists we reject $G$.

For $x \in N_1 \cup N_2$, let $N_K(x)$ denote the probes whose rows in $M^+_K$
intersect $x$'s row.  That is, $N_K(x)$ is the neighbors of $x$ given by $M^+_K$
when it is interpreted as a probe interval model.
There may be some vertices in $N(x) \setminus N_K(x)$.
These are the {\em unfulfilled adjacencies}.  They impose additional
constraints on the consecutive-ones orders of $M^+_K$, 
that allow semi-clique columns to be added to it to represent the unfulfilled
adjacencies. From Theorem~\ref{thm:invariantCols} it follows that every normal model
of $G - N_S$ has some ordering of columns of $M^+_K$ as a submatrix.

\subsection{Finding $M'_K$}\label{sect:binding}

Not every consecutive-ones ordering of $M^+_K$ is a submatrix of a normal model
of $G - N_S$.  To find such an ordering, we add constraint rows
to $M^+_K$, and find a consecutive-ones ordering $M'_K$ of the resulting matrix
to reflect the constraints imposed by the constraint rows.  This yields
an ordering $M^*_K = M'_K[V \setminus N_S]$ of $M^+_K$ that is a submatrix
of a normal model of $G - N_S$.

There are two types of constraints that must be reflected in $M'_K$,
{\em non-probe - probe binding constraints} and {\em probe - probe binding constraints}.

\subsubsection{Non-Probe - Probe Binding Constraints}\label{sect:NP-P}

Let $x \in N_1$ and let $p \in N(x) \setminus N_K(x)$. 
We know that $Q(x) \cap Q(p) = \emptyset$, because 
$p \notin N_K(x)$. 
Since $x$ and $p$ are 
adjacent, we know that their intervals must intersect in any model 
of $G - N_S$, and therefore $Q(x) \cup Q(p)$ must be consecutive
in $M^*_K$.
Let us call this additional constraint a {\em non-probe - probe binding 
constraint} imposed by $x$ and $p$.  We can enforce this constraint
by adding a new row equal to $Q(x) \cup Q(p)$ to $M^+_K$.
(See Figure~\ref{fig:constraints}.)

\begin{figure}
\centerline{\includegraphics[]{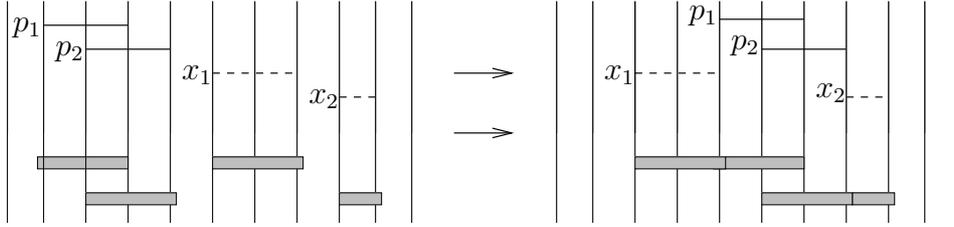}}
\caption{Enforcing binding constraints.  There is one column for each clique
of $G[P]$.  For each vertex $v$, $Q(v)$ is a row, and if
$G$ is a probe interval graph, this matrix has a consecutive-ones
ordering.  The solid lines depict the rows for two of the probes,
and the dashed lines depict the rows for two of the non-probes,
where $p_1$ is a neighbor of $x_1$ and $p_2$ is a neighbor of $x_2$
in $G$.  (Rows for other vertices are not depicted.)
These are {\em unfulfilled adjacencies}.
If $G$ is a probe interval graph, there is an ordering of
the columns where $p_1$ and $x_1$ are contained in adjacent columns and
and $p_2$ and $x_2$ are contained in adjacent columns.  This allows a new
column to be inserted between $p_1$ and $x_1$ where they can meet,
and similarly for $p_2$ and $x_2$.  These are {\em non-probe - probe
binding constraints}.  These constraints can be imposed on the
ordering of columns 
by inserting $Q(p_1) \cup Q(x_1)$ and 
$Q(p_2) \cup Q(x_2)$ as new rows and finding a consecutive-ones
ordering of the resulting matrix (shaded boxes).}\label{fig:constraints}
\end{figure}

Adding such a constraint for every such $x$ and 
$p$ will make $M$ too large for our time bound.
We show that a set of new rows with a linear
number of 1's is enough to enforce the non-probe - probe binding
constraints.

Let $x$ be a non-probe of $N_1$, and
let $c_i$ be the leftmost column of $Q(x)$ and let $c_j$
be the rightmost in (the yet unknown matrix) $M^*_K$.  We can divide $N(x) \setminus N_K(x)$ into the set
$Y_1$ of members that lie in columns to the left of $c_i$
and the set $Y_2$ that lie in columns to the right of $c_j$.
For each $p \in Y_1$, the rightmost column of $Q(p)$ is $c_{i-1}$;
the only way for $x$ and $p$ to be adjacent is to meet at a semi-clique
column between $c_{i-1}$ and $c_i$ in a normal model $M_N$.
Similarly, for each $p' \in Y_2$, the leftmost column of $Q(p')$
is $c_{j+1}$.  
No element of $Y_1$ is adjacent to any element of $Y_2$,
and $Y_1$ and $Y_2$ each induce complete subgraphs, since they
will contain a common endpoint of $x$ in $M_N$.
This implies that the same
$Y_1$ and $Y_2$ arise for $x$ in every normal model of $G - N_S$
(up to interchange).

Recall that the vertices are numbered from 1 through 
$n$. For two vertices $v$ and $u$, let $v \prec u$ denote that 
either $Q(v) \subset Q(u)$ or that $Q(v) = Q(u)$ 
and $v$ has a smaller vertex number than $u$ does; the numbers
serve as tie breakers.  Let $u \preceq v$ denote that $u \prec v$
or that $u = v$.

Since the members of $Y_1$ all end at the column to the left 
of $x$'s left endpoint and they all occupy consecutive cliques, it 
follows that for any $p,p' \in Y_1$, either 
$Q(p) \subseteq Q(p')$ or $Q(p') \subseteq Q(p)$.  
$Y_1$ induces a linear order in the $\prec$ relation.
It has a unique a minimal member $q$ in this relation. 
For example, for $x_2$ in Figure~\ref{fig:minConstraints},
$Y_1 = \{p_1, p_2, p_3\}$, $p_3 \prec p_2 \prec p_1$,
and $p_3$ is the unique minimal member of $Y_1$ in the $\prec$ relation.
Similarly, 
$Y_2$ has a unique minimal member $q'$ in the $\prec$ relation.
Also, $Q(q) \cap Q(q') = \emptyset$ since $q$
and $q'$ must lie on opposite sides of $x$.

By similar reasoning, each for each probe $p$,
the $\prec$ relation on non-probes that $p$ is bound to has at
most two nonadjacent minimal members $x_1$ and $x_2$.
Let us say that $x$ and $p$ are a {\em representative bound pair} if $p$ is
a minimal bound neighbor of $x$ and $x$ is a minimal bound neighbor of $p$
in the $\prec$ relation.  For example, in Figure~\ref{fig:minConstraints},
the two representative bound pairs are $\{p_3,x_2\}$ and $\{p_2,x_3\}$.

\begin{figure}
\centerline{\includegraphics[]{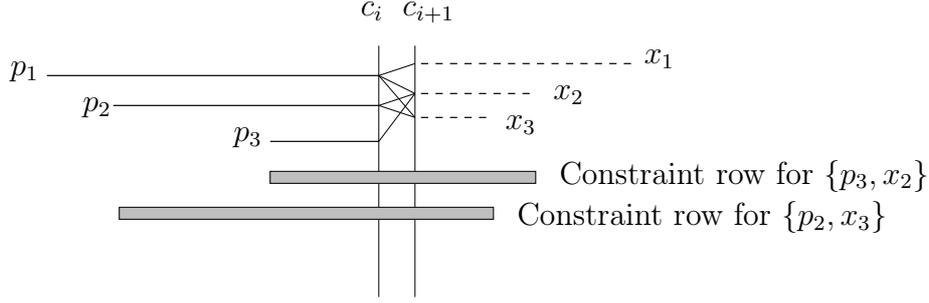}}
\caption{Representative bound pairs.  Probe $p_1$ is a neighbor of
non-probes $x_1$, $x_2$, and $x_3$; probe $p_2$ is a neighbor of $x_2$
and $x_3$, and probe $p_3$ is a neighbor of $x_2$.  That is six
binding constraints.  Since
all three probes share columns and all three non-probes share
columns, then if $G$ is a probe interval graph, the binding constraints
must force all three probes to occupy a column $c_i$ adjacent
to a column $c_{i+1}$ occupied by the non-probes.  The bound neighbor
of $x_2$ that minimizes the number of columns in this set is $p_3$
and the bound neighbor of $p_3$ that minimizes the number of columns
is $x_2$; they are each minimal to the other, so $\{p_2,x_2\}$ is a 
representative bound pair.
Similarly, $\{p_2,x_3\}$ is a representative bound pair.
Placing a constraint row only for each representative bound pair enforces
all binding constraints, since the minimal bound vertices $p_3$
and $x_3$ in the two sets each participate in a representative
bound pair.  Inserting one constraint row for each representative bound
pair in the entire matrix adds $O(n+m)$ 1's, since each vertex can participate
in at most two representative bound pairs, one at each of its endpoints,
and thereby contributes its 1's at most twice to constraint rows.
}\label{fig:minConstraints}
\end{figure}

We augment $M^+_K$ by adding a row for any representative pair $\{x, p\}$ that 
has $1$'s in the columns of $Q(x) \cup Q(p)$.   We show
below that doing this for all representative pairs adds $O(n+m)$ 1's
to the matrix, and enforces all non-probe - probe constraints,
not just those for representative bound pairs.

Let us now consider how to find the representative pairs
when we allow for the possibility that
$G$ is not a probe interval graph.  We have the endpoints
of each vertex in $M^+_K$.
First, for each $v \in P \cup N_1$ , we find its minimal bound neighbors 
in time proportional $|N(v)|$.
We create a list of the unfulfilled neighbors of $v$, and let $w_1$ be an arbitrary
neighbor of $v$ in the list.
For each unfulfilled neighbor $u$ in the list, we use the endpoints 
of the intervals corresponding to $Q(u)$ and $Q(w_1)$ to check whether
$Q(u)$ is disjoint from $Q(w_1)$, in which case it is not in $Y_1$,
or contains $Q(w_1)$, in which case it is in $Y_1$ and we eliminate it from
the list, or it is contained on $Q(w_1)$, in which case $u$ is in $Y_1$, we
eliminate it from the list, and let $w_1 = u$.
This takes $O(1)$ time per neighbor.
If a fourth case occurs, we reject $G$, since a necessary
condition described above is not met.  Otherwise, this
gives us $Y_1$ and the minimal neighbor in it for $v$.
If any unfulfilled neighbors remain, we find $Y_2$ and the minimal
member in it $w_2$ for $v$ in a similar way. If the list remains
nonempty following this, then $G$ is not a probe interval graph. 
  Summarizing:

\begin{lemma}\label{lem:Y1Y2}
In $O(n+m)$ time, we can either reject $G$, or find, for each $v \in P \cup N_1$, 
a partition of the unfulfilled neighbors of $v$ in $P \cup N_1$ into at most two 
sets, $Y_1$, $Y_2$, and label $v$ with $w,w'$,
such that $w$ is the only minimal member of $Y_1$,
$w'$ is the only minimal member of $Y_2$,
and $Q(w) \cap Q(w') = \emptyset$.
\end{lemma}

If we do not reject $G$, we have labeled each
vertex $v$ with at most two minimal bound neighbors
two minimal bound neighbors $w$ and $w'$.  We identify $\{v,w\}$ 
as a representative pair if $v$ is also one of the minimal bound
neighbors of $w$.  Similarly, we identify $\{v,w'\}$ if as a representative 
bound pair of $v$ is also a minimal bound neighbor of $w'$.
Since each vertex has at most two bound neighbors, there
are $O(n)$ pairs to perform this test on.

This gives the following:

\begin{lemma}\label{lem:repPairs}
In $O(n+m)$ time, we can either reject $G$, or find
the representative pairs for the non-probe - probe binding
constraints.
\end{lemma}

\subsubsection{Probe - Probe Binding Constraints} \label{sect:PP}

Consider $x \in N_2$. In this case, $Q(x) = \emptyset$ hence
$N_K(x) = \emptyset$.  All neighbors of $x$ are unfulfilled.
Then $x$ and its adjacencies must be represented exclusively by semi-clique columns.

Since $x \in N_2$, $N(x)$ does not induce a complete subgraph in $G$.
Let $p$ and $p'$ be two neighbors that are nonadjacent to each other.
$Q(p) \cap Q(p') = \emptyset$.  Their intervals must intersect $x$'s
in any model $M_N$ of $G = N_S$.  Therefore, $Q(p) \cup Q(p')$ must
be consecutive in $M^*_K$, so that semi-clique columns containing $x$ and $p$,
and $x$ and $p'$, can be placed in between $Q(p)$ and $Q(p')$.  
This is a {\em probe - probe binding constraint}.

Again, however, doing this for all such pairs $\{p,p'\}$ 
might add more
than $O(n+m)$ 1's.  It again suffices to add such rows for
only a subset of such pairs $\{p,p'\}$.

In a model of $G - N_S$, the 1's in row $x$
therefore lie between two consecutive clique columns $c_{i-1}$ and $c_i$.
Suppose that $c_{i-1}$ lies to the left of $c_i$. 
Let $Y_1 = N(x) \setminus S(c_i)$ and let $Y_2 = N(x) \setminus S(c_{i-1})$. 
The sets $Y_1$ and $Y_2$ satisfy $Y_1 \subseteq S(c_{i-1})$, $Y_2 \subseteq S(c_i)$ 
and $Y_1 \cap Y_2 = \emptyset$. Also, since $x$ is not simplicial, 
neither $Y_1$ nor $Y_2$ is empty.  Although we 
used a specific model to define $Y_1$ and $Y_2$ for $x$, these 
sets are unique for every $x \in N_2$, up to interchange between the two.

As with the non-probe - probe binding constraints, for each non-probe $x \in N_2$,
$Y_1$ and $Y_2$ are ordered by the $\prec$ relation.
We find the minimal members $p$ and $p'$ of $Y_1$ and $Y_2$
in the $\prec$ relation and make them {\em bound partners}.

For each probe, the bound partners can also be partitioned
into at most two sets that are ordered by the $\prec$ relation,
for the same reasons.
A {\em representative pair} is two bound partners that are each
minimal in the $\prec$ relation over bound partners of the other.

Let us now consider how to find the representative pairs
when we allow for the possibility that $G$ is not a probe
interval graph.
We apply the algorithm from the end of the previous section to verify
that for each
$x \in N_2$, $x$ has at most two minimal bound neighbors 
in the $\prec$ relation.
Similarly, for each probe $p$, we verify that $p$ has at most
two minimal bound partners in the $\prec$ relation.
We reject $G$ if these conditions do not apply, as we have
shown that they are necessary.  

The number of bound partners of any probe $p$ is bounded by the number
of neighbors in $N_2$, so it is $O(|N(p)|)$.  This is $O(m)$
over all probes.  We assign at most two minimal bound partners $q$ and $q'$ to
each probe $p$ in $O(n+m)$ time.  For $q$, we check whether $p$
is also one of its two minimal bound partners, and include $\{p,q\}$
as a representative pair if it is, and similarly for $\{p,q'\}$.

Proceeding as in the case of non-probe - probe constraints
gives us analogues of Lemmas~\ref{lem:Y1Y2} and~\ref{lem:repPairs}.
 
\begin{lemma}\label{lem:Y1Y2-2}
In $O(n+m)$ time, we can either reject $G$, or find, for each $v \in P$, 
a partition of the bound partners of $v$ into at most two sets, $Z_1$, $Z_2$, 
and label $v$ with $w_1,w_2$,
such that $w_1$ is the only minimal member of $Z_1$ in the $\prec$ relation,
$w_2$ is the only minimal member of $Z_2$ in the $\prec$ relation,
and $Q(w) \cap Q(w') = \emptyset$.
\end{lemma}

A representative pair is two bound partners that are each a minimal
bound partner of the other.  Since each vertex has at most two
bound partners, there are $O(n)$ pairs to test for this condition.

\begin{lemma}\label{lem:repPairs-2}
In $O(n+m)$ time, we can either reject $G$, or find
the representative pairs for the probe - probe binding
constraints.
\end{lemma}

\subsection{Finding $M^*_K$}

Let $M'_K$ be a consecutive-ones ordering of the matrix obtained by adding
$Q(u) \cup Q(v)$ as a row to $M^+_K$
for each non-probe - probe or probe - probe representative pair $\{u,v\}$, and 
let $M^*_K = M'_K[V \setminus N_S]$.

\begin{lemma}\label{lem:MPTime}
It takes $O(n+m)$ time either to reject $G$ or to compute a consecutive-ones
ordering $M^*_K$ of $M^+_K$ that observes all non-probe - probe and probe - probe constraints,
whether or not $G$ is a probe interval graph.
\end{lemma}
\begin{proof}
We find the representative pairs $\{u,w\}$ in $O(n+m)$ time by Lemmas~\ref{lem:repPairs}
and~\ref{lem:repPairs-2}.  We insert $Q(u) \cup Q(w)$ as a new row to $M^+_K$.
Since each vertex $v$ is a member of at most two non-probe - probe representative
pairs by Lemma~\ref{lem:Y1Y2}, and, similarly, at most two probe - probe representative
pairs, by Lemma~\ref{lem:Y1Y2-2}, it contributes at most $4|Q(v)|$ 1's to the new rows.
Over all $v$, that is $O(n+m)$ 1's.

If the resulting matrix does not have a consecutive-ones ordering, then
it is impossible to satisfy the probe - probe and probe - non-probe
binding constraints for even the representative pairs, and we can reject $G$.

Otherwise, we find a consecutive-ones ordering $M'_K$ of it
in $O(n+m)$ time, since the matrix has $O(n+m)$ 1's.
Let $M^*_K = M'_K[V \setminus N_S]$.  All conditions of the lemma but the 
last are now immediate.

Suppose a pair $u,v$ is bound by non-probe - probe constraint.  We show
that $Q(u) \cup Q(v)$ is consecutive in $M^*_K$ even if $\{u,v\}$ is not
a representative pair.
If $Q(u) \cup Q(v)$ is a constraint row of $M'_K$, then this
is immediate.
Otherwise, of the two minimal bound neighbors of $v$,
let $u'$ be one such that $u' \preceq u$.  
Of the two minimal bound neighbors of $u$, let $v'$ be the one
such that $v' \preceq v$.  The existence of $u'$ and $v'$ follows
from Lemma~\ref{lem:Y1Y2}.  Since $Q(u) \cup Q(v)$
is not a constraint row, $Q(v') \subset Q(v)$
or $Q(u') \subset Q(u)$.  Suppose without
loss of generality that $Q(v') \subset Q(v)$.
We may assume by induction on the number of 1's in the constraint
that $Q(v') \cup Q(u)$ is consecutive in $M'_K$.
Since $Q(v') \subset Q(v)$ and $Q(v)$ is disjoint
from $Q(u)$, $Q(u) \cup Q(v)$ is also consecutive
in $M'_K$.

All non-probe - probe binding constraints are satisfied in $M'_K$.
Similarly, if the pair $p,q$ is bound by a probe - probe constraint, $Q(p) \cup Q(q)$
is consecutive.  The proof is identical, except that it is applied to bound partners.
All probe - probe binding constraints are satisfied in $M'_K$, hence in $M^*_K$.
\end{proof}

\subsubsection{Sufficiency of the constraints}\label{sect:sufficiency}

By Theorem~\ref{thm:allModels}, the orderings of columns of
(the yet unknown matrix) $M_N$ given
by $T(M_N)$ are all normal models of $G - N_S$.  Since the columns of $M^+_K$
are the set of clique columns, $C_K$, of every normal model, it follows
that $T(M_N)[C_K]$ gives the set of orderings of columns of $M^+_K$ that are
submatrices of normal models of $G - N_S$.  We prove that the non-probe - probe
and probe - probe constraints are sufficient by showing that
$T(M'_K) \equiv T(M_N)[C_K]$.

The following algorithm is not meant to be efficient; it is a tool for proofs.  
(See Figure~\ref{fig:colDelete}.)

\begin{figure}
\centerline{\includegraphics[]{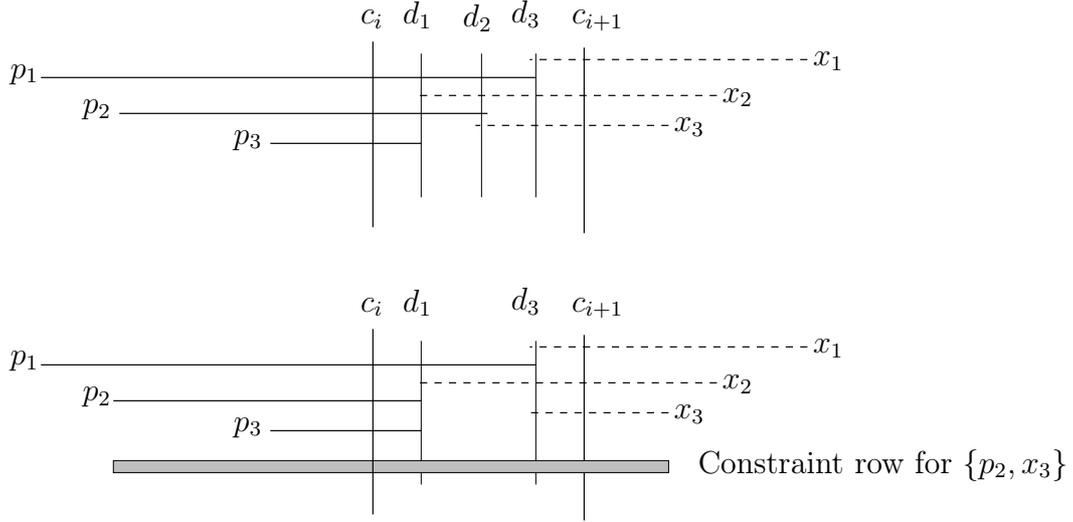}}
\caption{The action of Algorithm~\ref{alg:columnDelete} when semi-clique columns
are deleted from a normal model of $G - N_S$.  When the algorithm deletes
column $d_2$, it adds a constraint row for each pair $(\ell,r)$, where
$\ell$ has its right endpoint and $r$ has its left endpoint in the deleted
column.  In this case, this is just $(p_2,x_3)$.  This simulates the effect
of restricting the PQ tree.
When all of $d_1, d_2, d_3$ are deleted, the constraint rows it has added
are just the binding constraints for Figure~\ref{fig:minConstraints},
not just the ones for representative pairs.  
Doing this for all semi-clique columns shows that $T(M'_K) \equiv T(M_N)[C_K]$,
which means that the binding constraints are sufficient to make $M^*_K$
a submatrix of a normal model $M_N$ of $G - N_S$.
}\label{fig:colDelete}
\end{figure}

\begin{algorithm}\label{alg:columnDelete}

Given a consecutive-ones ordered clique matrix $M$,
delete a column $c$ from $M$ 
and add new rows so that for the resulting matrix 
$M'$, $T(M') \equiv T(M) - c$.

{\bf Precondition:}  {\em Column $c$
contains both a proper right endpoint and a proper left endpoint in $M$}

\begin{itemize}
\item Let $L$ be the rows with proper right endpoints in 
$c$ and let $R$ be the rows with proper left endpoints in $c$.  For 
each element $(\ell,r)$ of $L \times R$,
insert $S(\ell) \cup S(r)$ as a {\em constraint row}.  

\item Let $M''$ be the result.  Delete column $c$ from $M''$, yielding $M'$.
\end{itemize}
\end{algorithm}

\begin{lemma}\label{lem:columnDelete}
Algorithm~\ref{alg:columnDelete} is correct.
\end{lemma}
\begin{proof}
When rows $X$ and $Y$ of $M$ properly overlap, then since they are each consecutive,
so is $X \cup Y$, by Lemma~\ref{lem:overlap}.
Therefore, adding $X \cup Y$ as a row to the matrix does nothing
to the PQ tree of the matrix.  $T(M'') \equiv T(M)$.  
It suffices to show that $T(M'')- c \preceq T(M')$ and $T(M') \preceq T(M'') - c$.
That $T(M'') - c \preceq T(M')$ follows from Lemma~\ref{lem:RestrictPrec}.

By Lemma~\ref{lem:FPerp}, to show $T(M') \preceq T(M'') - c$,
it suffices to show that for every $B \in {\cal F}(M')$,
either $B \in {\cal F}(M'')$ or $B \cup \{c\} \in {\cal F}(M'')$.
To obtain a contradiction, suppose that this is not true
for some $B \in {\cal F}(M')$.

This implies $B \not\in {\cal F}(M'')$.
The set $B$ properly overlaps
some row $S(v)$ of $M''$, but $B$ does not properly overlap $S(v) - c$.
In other words, $B \setminus S(v)$, $B \cap S(v)$, and $S(v) \setminus B$
are all nonempty, but one of $B \setminus (S(v) - c)$, $B \cap (S(v) - c)$,
and $(S(v) - c) \setminus B$ is empty.  Since, $B$ is a set of columns of $M'$,
 $c \not\in B$, so $B \setminus (S(v) - c) = B \setminus S(v)$ and $B \cap (S(v) - c) = B \cap S(v)$,
and these are nonempty.
Only $(S(v) - c) \setminus B$ is empty, and since $S(v) \setminus B$ is nonempty,
$S(v) \setminus B = \{c\}$.

Both of $B$ and $S(v)$ are consecutive in $M''$.
Suppose without loss of generality that the left endpoint of $B$
is farthest to the left.  Then $c$ is the right endpoint of $v$
and it is a proper right endpoint.
There exists a row $w$ of $M''$ with a proper left endpoint in $c$
by the precondition of Algorithm~\ref{alg:columnDelete}.
Algorithm~\ref{alg:columnDelete} inserted
$S(v) \cup S(w)$ as a row of $M''$, and $(S(v) \cup S(w)) - c$ is a row of $M'$
that properly overlaps $B$, contradicting $B \in {\cal F}(M')$.
\end{proof}

Recall that $M_N$ is a normal model of $G - N_S$, and let $M_J$ be the submatrix given by its clique
columns.  Note that $M_J$ is just a consecutive-ones ordering of columns of $M^+_K$.
Let $M'_J$ be the result of adding all non-probe - probe and probe - probe constraints
as rows to $M_J$, not just the ones given by representative pairs.  

\begin{lemma}\label{lem:constraintsSufficient}
$T(M'_K) \equiv T(M'_J) \equiv T(M_N)[C_K]$.
\end{lemma}

\begin{proof}
Iteratively applying Algorithm~\ref{alg:columnDelete} to non-clique
columns of $M_N$ in any order leaves the columns of $M_J$, but adds rows,
yielding a matrix $M''_J$.
By Lemma~\ref{lem:columnDelete} and
induction on the number iterations, $T(M''_J) \equiv T(M_N)[C_K]$.

To show $T(M'_J) \equiv T(M_N)[C_K]$, we show that $T(M''_J)$ is the
result of adding one constraint row realizing each probe - probe and probe - probe constraint
to $M_J$.  Since we have shown that the rows for representative pairs added to $M_J$ to obtain
$M'_J$ realize these constraints, and they are a subset of the rows added
to $M''_J$, the result will follow.

Suppose $A_1$ is the initial set of columns of $M_N$, and that
$\{x,p\}$ have a non-probe - probe binding constraint in $M_J$.
Suppose without loss of generality that $k_i$
is the rightmost clique column in $S(p)$ and $k_{i+1}$ is the leftmost in $S(x)$.
The constraint means that $S(p)$ and $S(q)$ meet at a semi-clique column between $k_i$ and $k_{i+1}$.
Let $A_2$ be the set of columns just after the last column in $S(p) \cap S(x)$
is deleted.  At that time, Algorithm~\ref{alg:columnDelete} has just added $(S(p) \cap A_2) \cup (S(x) \cap A_2)$
as a new constraint row.  When only clique columns remain, this row is
$(S(p) \cap C_K) \cup (S(x) \cap C_K) = Q(p) \cup Q(x)$, the non-probe - probe binding constraint
for $p$ and $x$.

Suppose $p$ and $q$ are probes that will have a probe - probe binding constraint
in $M_J$.  Suppose without loss of generality
that $k_i$ is the rightmost clique column
of $S(p)$ and $k_{i+1}$ is the leftmost clique column of $S(q)$.  The constraint
means that $p$ and $q$ are not neighbors, but that they
have a common neighbor $x \in N_2$.  Since $x \in N_2$, $S(x)$ does not contain
a clique column, so $S(x)$ is confined
to the columns of $M_N$ between $k_i$ and $k_{i+1}$.
Let $A_3$ be the columns that remain right after the last column in $S(p) \cap S(x)$
is deleted and $A_4$ be the columns that remain right after the last column
in $S(x) \cap S(q)$ is deleted.  Assume without loss of generality that the last column in 
$S(p) \cap S(x)$ is deleted first.
When $A_3$ remains, Algorithm~\ref{alg:columnDelete} inserts $(S(p) \cap A_3) \cup (S(x) \cap A_3)
= (S(p) \cup S(x)) \cap A_3$ as a new row, $R$.
When $A_4$ remains, Algorithm~\ref{alg:columnDelete} inserts $(R \cap A_4) \cup (S(q) \cap A_4)$
as a new row.  This is equal to $[(S(p) \cup S(x)) \cap A_4] \cup (S(q) \cap A_4) = 
(S(p) \cup S(x) \cup S(q)) \cap A_4$.
What remains of this row after the column set is $C_K$ is $(S(p) \cup S(x) \cup S(q)) \cap C_K =
(S(p) \cup S(q)) \cap C_K = Q(p) \cup Q(q)$, the probe - probe binding constraint
for $p$ and $q$.
\end{proof}

\subsection{Adding columns to $M^*_K$ to obtain a normal model $M_N$ of $G - N_S$}

We now know by Lemma~\ref{lem:constraintsSufficient} that $M^*_K$
is a submatrix of a normal model of $G - N_S$.
The next lemma describes the structure of the columns 
that must be inserted between each pair $\{c_i$,$c_{i+1}\}$ of consecutive 
columns of $M^*_K$ to obtain a normal model $M_N$ of $G - N_S$.
(See Figure~\ref{fig:Cxy}.)

\begin{definition}
A sequence $(S_1, S_2, \ldots, S_k)$ of sets is {\em ascending}
if $S_i \subset S_{i+1}$ for each $i$ such that $1 \leq i < k$ and {\em descending}
if $S_{i+1} \subset S_i$ for each $i$ such that $1 \leq i < k$.
\end{definition}

\begin{lemma}\label{lem:descendingAscending}
Let $(c, d_1, d_2, \ldots, d_k, c')$ be a consecutive set of columns, in left-to-right order,
in a normal model $M_N$ of $G - N_S$, such that $c$ and $c'$ are clique columns
and for each $i$ such that $1 \leq i \leq k$, each $d_i$ is a semi-clique column.

\begin{itemize}
\item The columns of $\{d_1, d_2, \ldots, d_k\}$ whose probe sets are subsets of $S(c)$
precede the columns whose probe sets are subsets of $(c')$.
Let $(d_1, d_2, \ldots, d_h)$ be the ones whose probe sets
are subsets of $S(c)$.

\item The probe sets $(S(c) \cap P, S(d_1) \cap P, S(d_2) \cap P, \ldots, S(d_h) \cap P)$ are a descending
sequence.

\item The probe sets 
$(S(d_{h+1}) \cap P, S(d_{h+2}) \cap P, \ldots, S(d_k) \cap P, S(c') \cap P)$ are an ascending sequence.
\end{itemize}
\end{lemma}
\begin{proof}
No probe in $S(c_i) \setminus S(c_{i+1})$ can meet a probe in $S(c_{i+1}) \setminus S(c_i)$
in any column of $\{d_1, \ldots, d_k\}$, since this would be a clique
column between $c$ and $c'$, a contradiction.   No column can have
$(S(c) \cap P) \cap (S(c') \cap P)$ as its probe set, since then any endpoint in
the column fails to be taut, or it has no endpoints, contradicting the normality of $M_N$ in
either case.
The columns can be uniquely partitioned into $\{d_1, \ldots, d_h\}$ and $\{d_{h+1}, \ldots, d_k\}$
such that the probe sets in the first of these are subsets of $S(c) \cap P$ but not of $S(c') \cap P$,
and the probe sets in the second are subsets of $S(c') \cap P$ but not of $S(c) \cap P$.
Since $M_N$ is normal, no probe sets in consecutive columns are equal,
by Lemma~\ref{lem:nextProbes}.
Therefore, consecutiveness of 1's then forces
$(S(c) \cap P, S(d_1) \cap P, S(d_2) \cap P, \ldots, S(d_h) \cap P)$ to be a decreasing sequence.
By symmetry, $(S(d_{h+1}) \cap P, S(d_{h+2}) \cap P, \ldots, S(d_k) \cap P, S(c') \cap P)$ is 
an ascending sequence.
\end{proof}

\begin{figure}
\centerline{\includegraphics[]{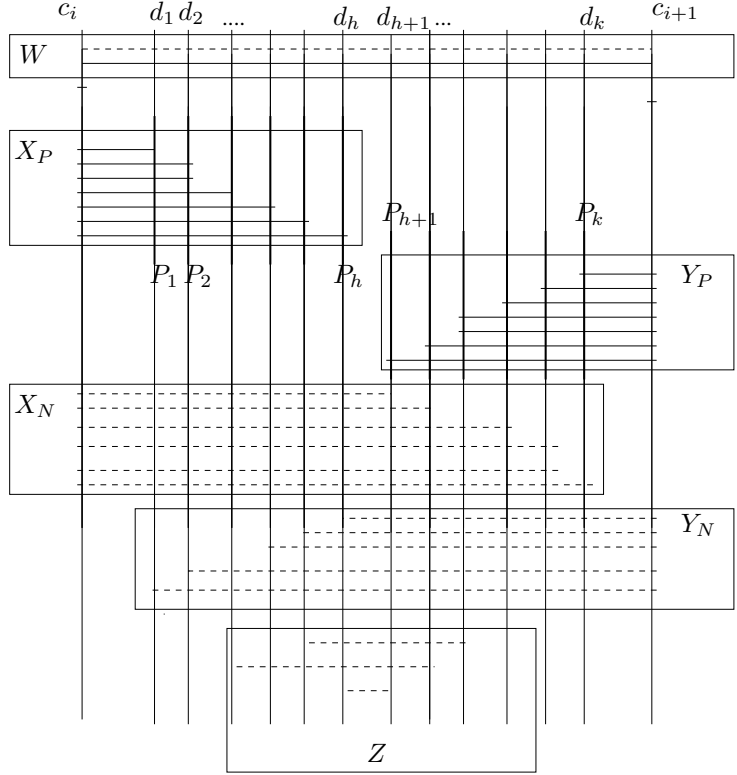}}
\caption{The structure of columns $(d_1, d_2, \ldots, d_k)$ inserted between two
columns $c_i$ and $c_{i+1}$ of $M^*_K$ in obtaining $M_N$.
Solid lines are probes; dashed lines are
non probes.  $W = W(i) = S(c_i) \cap S(c_{i+1})$.  $X_P = X(i) \cap P$ and $X_N \ X(i) \cap N$ are the probes
and non-probes, respectively,
that have right endpoints at $c_i$, but also
some neighbors in $c_{i+1}$.  
$Y_P$ and $Y_N$ are defined symmetrically.
$X(i) = X_P \cup X_N$ and $Y(i) = Y_P \cup Y_N$ are the elements with unfulfilled
adjacencies that must be represented by insertion of $(d_1, d_2, \ldots, d_k)$,
and and $Z(i)$ is the members of $N_2$ 
whose unfulfilled adjacencies with
probe neighbors in $c_{i}$ and
and probe neighbors with left endpoints in $c_{i+1}$ that must also be
represented by these columns.
}\label{fig:Cxy}
\end{figure}

\begin{definition}\label{def:newStretchers}
Let $c_i$ and $c_{i+1}$ be consecutive columns in $M^*_K$.
Let $Z(i)$ be the members of $N_2$ that have
neighbors in both $S(c_i) \setminus S(c_{i+1})$
and in $S(c_{i+1}) \setminus S(c_i)$.
Let $X(i)$ be the set of vertices in $S(c_i) \setminus S(c_{i+1})$ that have
neighbors in $Z(i) \cup S(c_{i+1}) \setminus S(c_i)$.
Let $Y(i)$ be the set of vertices in $S(c_{i+1}) \setminus S(c_i)$ that have
neighbors in $Z(i) \cup S(c_i) \setminus S(c_{i+1})$.
\end{definition}

Clearly, the unfulfilled adjacencies that must be represented by adding new rows
between $c_i$ and $c_{i+1}$ are those adjacencies between members
of any two of $\{X(i),Y(i),Z(i)\}$.

\begin{lemma}\label{lem:stretchEndpoints}

Let $\{X(i),Y(i),Z(i)\}$ be as in Definition~\ref{def:newStretchers}.
Let $M_N$ be a normal model of $G - N_S$
such that $M^*_K$ is a submatrix of $M_N$.
For each non-probe $x$ in $Y(i) \cup Z(i)$, the probe set of $x$'s left
endpoint in $M_N$ is $N(x) \cap S(c_i)$.  By symmetry, the probe set of the right endpoint
of each non-probe $x'$ in $X(i) \cup Z(i)$ is $N(x') \cap S(c_{i+1})$.
\end{lemma}
\begin{proof}
Immediate from Lemma~\ref{lem:descendingAscending}.
\end{proof}

Lemma~\ref{lem:stretchEndpoints}
gives probe sets of columns that must be inserted between clique columns
$c_i$ and $c_{i+1}$.  No other probe
set can occur in them: any endpoints in the column would not
be taut, and if the column has no endpoints, it would be a subset
of another, contradicting normality of $M_N$.  

Using the radix sorting
technique mentioned in Section~\ref{sect:initial}, we can find,
for each consecutive pair $\{c_i,c_{i+1}\}$ of columns of $M^*_K$, 
the probe sets of left endpoints in descending order of size,
and the probe sets of right endpoints in ascending order of size,
and we may eliminate duplicate copies of the same set.  This takes $O(n+m)$ time.
We reject $G$ if they do not form descending and ascending sequences,
as required by Lemma~\ref{lem:descendingAscending}, and this
also takes $O(n+m)$ time to check, as described in Section~\ref{sect:initial}.

If we have not rejected $G$, this identifies the unique order of the
columns containing these probe sets.  This gives the position of
left endpoint of every non-probe in $Y(i) \cup Z(i)$, and the
position of the right endpoint of every non-probe in $X(i) \cup Z(i)$.
For each $z \in Z(i)$, we must add 1's between the left and right endpoint
of $z$.  For each non-probe $x \in X(i)$, we must add 1's between $c_i$
and $x$'s right endpoint.  For each non-probe $y \in Y(i)$, we must add 1's
between $c_{i+1}$ and $y$'s left endpoint.  

This gives the members of $X(i)$, $Y(i)$, $Z(i)$ in each column, as well as the
probes in $W(i) = c_i \cup c_{i+1}$, which must also appear in each of the
new columns.  For each non-probe
$w \in W(i)$, we must add 1's between $c_i$ and $c_{i+1}$.
Since the order of probe sets satisfies the requirements of
Lemma~\ref{lem:stretchEndpoints}, this
fulfills the adjacencies between $X(i)$, $Y(i)$, and $Z(i)$.  

We cannot add any other 1's to these columns without violating the 
requirements of a normal model.  
The columns between $c_i$ and $c_{i+1}$
and their ordering is uniquely determined.  Therefore,
performing this operation at all pairs of columns of $M^*_K$,
we obtain a normal model of $G - N_S$,
which is uniquely determined, given $M^*_K$.  Since it is
a normal model, it has $O(n+m)$ 1's and we have spent $O(n+m)$ time.

\section{The Consecutive-Ones Probe Matrix Problem}\label{sect:C1PM}

Recall that an instance the consecutive-ones probe matrix
problem is a matrix $M$ whose elements
are 0's, 1's, and $\ast$'s, and the $\ast$'s
form a submatrix.  We seek to find a way to replace the $\ast$'s with
0's and 1's so that the resulting matrix has the consecutive-ones property.
We assume that the instance of the problem is given
by two matrices:  the submatrix $M_R$ consisting of those rows that do not
contain $\ast$'s, and the submatrix $M_C$ consisting of those columns that
do not contain $\ast$'s.  The columns of $M_C$ are a subset of the columns
of $M_R$ and the rows of $M_R$ are a subset of the rows of $M_C$.
Denote the set of row of $M_R$ by $R$ and the set of columns of $M_C$ by $C$.
Since the $\ast$'s form a submatrix, all entries that are 0 or 1 occur
in a row of $M_R$ or a column of $M_C$ (or both), while no $\ast$ appears
in either matrix.  By using sparse representations of $M_R$ and $M_C$,
we get a representation of the instance in space proportional to the number
of rows, columns, and 1's of $M$.

\begin{lemma}\label{lem:C1PMTest}
The consecutive-ones probe matrix problem on $M_R$ and $M_C$
has a solution if and only if there exists a consecutive-ones
ordering $\pi$ of columns of $M_R$ such that $\pi[C]$ is also
a consecutive-ones ordering of $M_C$.
\end{lemma}
\begin{proof}
Let $\pi$ be an ordering of columns of $M$ (hence of columns of $M_R)$
that makes it possible to fill in the $\ast$'s so that $M$
is consecutive-ones ordered.  Since each row of $M_R$ is a row
of $M$, $\pi$ must be a consecutive-ones ordering of $M_R$.
If $\pi[C]$ is not a consecutive-ones ordering of $M_C$, then
in some row of $M_C$, hence of any assignment of $\ast$'s
in $M$ has a 0 occurs between two 1's, a contradiction.  Therefore,
$\pi$ is a consecutive-ones ordering of $M_R$ and $\pi[C]$ is
a consecutive-ones ordering of $M_C$.

Let $\pi'$ be an ordering of columns of $M_R$ such that $\pi'[C]$
is also a consecutive-ones ordering of columns of $M_C$.
Then each row of $M$ in $R$ is consecutive-ones ordered.
In each row $y$ of $M$ that is not in $R$, let $\pi(c_1)$ and 
$\pi(c_2)$ be the first and last positions of columns of $C$
that have 1's in the row.  Then for every column $c_3 \in C$
such that $\pi(c_1) < \pi(c_3) < \pi(c_2)$, $\pi(c_3)$
has a 1 in row $y$ and for every column $c_4 \in C$ such that
$\pi(c_4) < \pi(c_1)$ or $\pi(c_4) > \pi(c_2)$, $c_4$ has a 0
in row $y$.  Setting any $\ast$'s between $\pi(c_1)$
and $\pi(c_2)$ to 1 results in a consecutive-ones ordered matrix
$M$.
\end{proof}

\begin{definition}\label{def:C1PM}
We will let $\CiPM(M_R, M_C)$ denote an instance of the consecutive-ones probe matrix problem,
where the columns of $M_C$ are a subset $C$ of the columns of $M_R$ and the rows of $M_R$
are a subset $R$ of the rows of $M_C$.  By Lemma~\ref{lem:C1PMTest}, a solution is any ordering 
$\pi$ of columns of $R$ such that $\pi$ is a consecutive-ones ordering $M^*_R$ of $M_R$ and $\pi[C]$ is
a consecutive-ones ordering $M^*_C$ of $M_C$.  The matrix $M$ obtained by assigning
$\ast$'s to 1 if and only if they occur between 1's in columns of $C$ is
the {\em taut} matrix implied by $\pi$.
\end{definition}

The reason for the distinguishing the taut matrix implied by a solution $\pi$ is that
$\pi$ does not always uniquely specify a required assignment of $\ast$'s.
Let $y$, be a row of $M$, and let $c$ be the first column of $M$ to
the left of the block of 1's in $y$.  Then the 0 assigned to a $\ast$ in row $y$, column $c$
in the taut matrix can be changed to a 1 without violating the constraints.
This can be iterated until the column to the left of the block of 1's is a column
in $C$.  Similarly, 0's after the rightmost 1 in
$y$ in a taut solution might be able to be reset to 1's.  Conceptually,
tautness is analogous to tautness in a probe interval model:
a solution $M$ is taut if no endpoint of a row can be set to 0 to
obtain a smaller solution consistent with the constraints.  When we
use it in the probe interval graph recognition algorithm, this allows
$M$ to be extended to a normal model, which must be taut.  By definition,
the implied taut matrix is unique for each solution $\pi$.

It is not necessary to construct a sparse representation of $M$ explicitly, and generally
it would not be possible to do this in time linear in the size of the inputs.
The $\ast$'s that must be 1 can greatly exceed the number of 1's in the inputs.
Given a solution $\pi$, we may create, in time linear in the size of the inputs 
a representation of $M$ that allows $O(1)$-time lookup of any entry.
It suffices to record, for each row, the position of the first and last 1 in the row.
If the row is in $R$, then this is the first and last position of a 1
in $M^*_R$, and if the row is not in $R$,  it is $\pi(c_1)$
and $\pi(c_2)$ for the first and last columns $c_1$ and $c_2$, respectively,
of $M^*_C$, that have $1$'s in the row.

This representation of $M$ takes space proportional to the sizes of $M_C$ and $M_R$.
For a row $i$ and column $j$, it 
takes $O(1)$ time to determine the value of the element at row $i$, column $j$
of $M$, by determining whether $j$ is in the interval between the first and 
last 1 of row $i$.

\begin{figure}
\centerline{\includegraphics[]{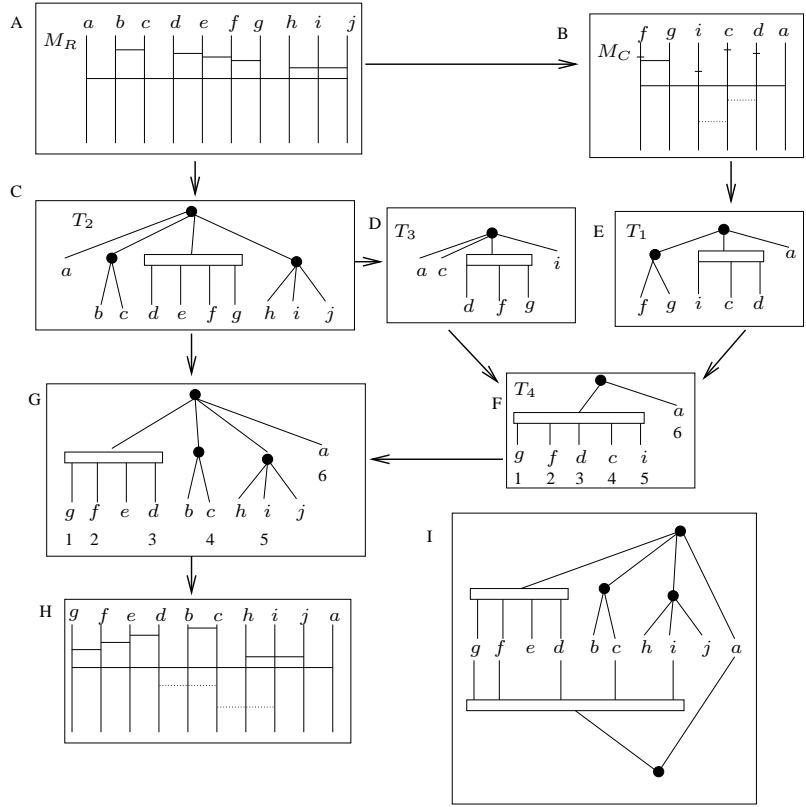}}
\caption{Solving the consecutive-ones probe matrix problem.  A:  $M_R = M[R]$ for an unknown
matrix $M$, consecutive-ones ordered.  Blocks of 1's are represented with line segments.
B:  $M_C = M[][C]$ for $M$, consecutive-ones ordered,
for column set $C = \{a,c,d,f,g,i\}$.  Rows of $M_C$ that are not in $M_R$ are dotted.   
Elements of $M$ that are neither in a row of $M_R$ nor a column of $M_C$ are implicitly $\ast$'s, and can be freely
assigned a value of 0 or 1 in $M$.  In this example, these are elements that are in
the two additional rows in $M_C$, depicted with 1's, and
in the columns $\{b,e,h,j\}$
that are in $M_R$ but not in $M_C$.
C:  The PQ tree $T_2$ of $M_R$.
D:  $T_3 \equiv T_2[C]$.  E: The PQ tree $T_1$ of $M_C$.  F:  $T_4 \equiv T_1 \cap T_3$,
which gives a left-to-right numbering of $C$ that is consistent with
both $T_1$ and $T_2$.  G:  A reordering of $T_2$ consistent with this numbering.
H:  The resulting ordering of columns of $M$, including those in $M_C$.  Since the
submatrix given by columns of $M_C$ is consecutive-one ordered, $\ast$'s lying between 1's in
these columns are assigned a value of 1, yielding a consecutive-ones ordering
of $M$.  I: A gadget suggested by the procedure.  $T_4$ (inverted) and $T_2$ share leaves.
Since $T_4 \preceq T_2[C]$, $T_4$ can be ordered without interference from $T_2$, and
once that has been done, $T_2$ can be ordered to place the remaining leaves.
}\label{fig:probeOnes}
\end{figure}

\begin{algorithm}\label{alg:C1PM}
Solve an instance $\CiPM(M_R, M_C)$ of the consecutive-ones probe matrix problem, or determine
that no solution exists.

\begin{enumerate}
\item Let $T_1$ be the PQ tree of $M_C$ and $T_2$ be the PQ tree of $M_R$.
Return that the problem has no solution if $T_1$ or $T_2$ does not exist.

\item Let $T_3 \equiv T_2[C]$.   (See Figure~\ref{fig:probeOnes}.)
Now $T_1$ and $T_3$ have the same leaf sets.  
Find $T_1 \cap T_3$ or return that the problem has no
solution if $T_1 \cap T_3$ does not exist.

\item Let $\tau \in \Pi(T_1 \cap T_3)$.  This gives a left-to-right
numbering of the subset $C$ of columns of $M_R$.

\item Number each leaf of $T_2$ that is in $C$ in ascending
order of this numbering.  Label each internal node of $T_2$
with a {\em descendant number}, the leaf number of any
descendant, if a numbered leaf descendant exists.
The descendant number of a leaf is its leaf number.

\item For each P node, sort the children in ascending order of descendant labels.
Place unlabeled children anywhere in the ordering.

\item For each Q node, if it has at least two children with labels,
orient the Q node so that the child with the smaller descendant
label is earlier.  Otherwise, choose one of the two possible
orderings arbitrarily.

\item Return the resulting ordering $\pi$ of leaves of $T_2$.
\end{enumerate}
\end{algorithm}

\begin{lemma}\label{lem:C1PMCorrect}
Algorithm~\ref{alg:C1PM} is correct.
\end{lemma}
\begin{proof}
Suppose an instance of the problem has a solution.  Then let $\pi$ be a solution.
$T_2$ exists, since $\pi$ is a consecutive-ones ordering of $M_R$, and $T_1$
exists, since $\pi[C]$ is a consecutive-ones ordering of $M_C$.
Thus, $\pi[C] \in \Pi(T_1)$ and $\pi \in \Pi(T_2)$, hence $\pi[C] \in \Pi(T_3)$.
Since $\pi \in \Pi(T_1) \cap \Pi(T_3)$, this set is nonempty, so
$T_1 \cap T_3$ is defined, and the algorithm correctly claims that the problem
has a solution.  

Conversely, suppose that the algorithm claims that the problem has a solution.
Then $T_1 \cap T_3$ is defined.  The algorithm finds
$\tau \in \Pi(T_1 \cap T_3)$.  This implies
$\tau \in \Pi(T_1)$, so $\tau$ is a consecutive-ones ordering of $M_C$.
Since $\tau \in \Pi(T_3)$,  $\tau = \pi'[C]$ for for some $\pi' \in \Pi(T_2)$.
Because $\pi' \in \Pi(T_2)$,
$T_2$ can be ordered so that leaf numbers are ordered
in increasing order of leaf number.  This gives a consecutive-ones
ordering of $M_R$, since $T_2$ is the PQ tree of $M_R$.  A solution exists.
For each node $u$ of $T_2$, let $i$ and $j$ be
the minimum and maximum leaf label assigned to leaf descendants
by the algorithm, using $\tau$, and let $[i,j]$ be $u$'s {\em interval}.
Since $\pi'$ exists, it is a valid ordering of the leaves of $T_2$.
Since $\tau = \pi'[C]$, the intervals of the children are disjoint, and for each 
Q node, the ordering of these intervals on the number line is consistent
with the ordering of children of the Q node, or its reverse.
The procedure for ordering $T_2$ will therefore produce an ordering
of children at each internal node where the intervals of children are
disjoint and consistent with their ordering on the number line.
By induction on the size of the subtree rooted at a node of $T_2$,
this gives an ordering $\pi'' \in \Pi(T_2)$ where numbered leaves
are in ascending order, hence where $\pi''[C] = \tau$.  The algorithm returns
a correct solution.
\end{proof}

\begin{lemma}\label{lem:C1PMTime}
An instance $\CiPM(M_R, M_C)$ of the consecutive-ones probe matrix problem can be solved in time
linear in the number of rows, columns and 1's in  $M_R$ and $M_C$.
\end{lemma}
\begin{proof}
Finding $T_1$ and $T_3$ takes linear time by the algorithm
of Booth and Lueker.  

Finding $M(T_1)$ and $M(T_3)$ takes time linear in the sizes
of $M_R$ and $M_C$, by Lemma~\ref{lem:MofT}.
$T_1 \cap T_3$, is the PQ tree of the matrix whose rows are the union of rows
of $M(T_1)$ and $M(T_3)$. It takes linear time to construct $T_1 \cap T_3$, or to determine
that it does not exist by the algorithm of Booth and Lueker.

If $T_1 \cap T_3$ exists, then
numbering the leaves of this tree and then using it
to label the descendant numbers of leaves of $T_2$ takes $O(n)$
time.  It takes $O(n)$ time to label the descendant numbers
of the internal nodes of $T_2$ in postorder, letting each
node inherit its label from a child, if it has a labeled child.

By numbering the P nodes,  we may sort the children of all P nodes
with a single radix sort that uses parent number as the primary sort key
and descendant number as the secondary sort key.  There are $O(n)$
nodes in the tree, so this takes $O(n)$ time.  
\end{proof}

\begin{definition}\label{def:equivalence}
Two different bijections $\pi$ from a set $X$ to $\{1, 2, \ldots, |X|\}$
are {\em equivalent} if they are equal or one is the reverse of the other.
Two normal models $M$ and $M'$ of $G$ are {\em equivalent} if they are equal matrices, or one can be
obtained by reversing the column order of the other.  The normal model
for $G$ is {\em unique} if it only has two (equivalent) normal models.
\end{definition}

\begin{algorithm}\label{alg:C1PMUnique}
Determine whether a solution $\pi$ to an instance $\CiPM(M_R,M_C)$ found by Algorithm \ref{alg:C1PM}
is unique.

\begin{itemize}
\item If $M_R$ has at most two columns, return true.

\item If the only internal node of $T_2$ is a Q node, return true.

\item If $M_C$ has at least three columns and $T_1 \cap T_3$ fails to have a single internal node
that is a Q node, return false.

\item If not all children of a P node of $T_2$ are labeled with a descendant
number, return false.

\item If fewer than two children of a Q node of $T_2$ are labeled with descendant
numbers, return false.

\item Else return true.
\end{itemize}
\end{algorithm}

Note that failing to have a unique solution does not imply that
the taut matrices implied by the solutions are not equal up to reversal of column order, since the
different solutions could all be automorphisms of two matrices, where one is the reversal of
column order of the other.  
We do not address this second notion of uniqueness in this paper, since we are dealing
with matrices where the columns are labeled.  

\begin{lemma}\label{lem:C1PMUnique}
Algorithm~\ref{alg:C1PMUnique} is correct.
\end{lemma}
\begin{proof}
Every 0-1 matrix with fewer than three columns is consecutive-ones ordered, since there
is no way for a 0 to appear between two 1's in a row.  If $M_R$ has fewer than three
columns, then so does $M_C$, and any of the two orderings of $M_R$ is a solution.

Suppose $M_R$ has at least three columns.
By Lemma~\ref{lem:C1PMCorrect}, $T_1 \cap T_3$ is defined, since we have assumed
that the instance has a solution.
If the only internal node of $T_2$ is a Q node, then since every
solution is an element of $\Pi(T_2)$, there are only two solutions,
and one is the reverse of the other.

If $M_C$ has at least three columns, then $T_1 \cap T_3$ has at least three leaves.
Unless its only internal node is a Q node, it admits two orderings that are not
the reverse of each other.  In this case, Algorithm~\ref{alg:C1PM} can
produce two solutions solutions that are not the reverse of each other.

Otherwise, $T_2$ has at least three leaves, but does not have a single internal
node that is a Q node.
If not all children of a P node of $T_2$ are labeled with a descendant
number, then there are choices about how to order the unlabeled children,
and if fewer than two children of a Q node are labeled with descendant
numbers, we may choose one of the two allowed orderings of the children.
Since $T_2$ has at least three leaves and it does not have a single internal
node that is a Q node,  
these choices give rise to solutions that are not equivalent.
If on the other hand these conditions are satisfied, there is a unique way
to order $T_2$ (up to reversal).
\end{proof}

\begin{lemma}\label{lem:C1PMUniqueTime}
Algorithm~\ref{alg:C1PMUnique} takes $O(n+m)$ time.
\end{lemma}
\begin{proof}
The time for the operations is bounded by the time for the operations of
Algorithm~\ref{alg:C1PM}, so the lemma follows by Lemma~\ref{lem:C1PMTime}.
\end{proof}

\section{Finding a normal model of $G$}\label{sect:finalStep}

Suppose $G$ is a probe interval graph.
Our strategy is to find a matrix $M_P$ that has a
consecutive-ones ordering $M^*_P$ equal to 
$M_G[P]$ for some normal model $M_G$ of $G$.
Each column of $M_P$ that is not a column of $M_N[P]$ is
the neighborhood of a member of $N_S$.
By Theorem~\ref{thm:invariantCols},
some consecutive-ones ordering $M^*_N$ of $M_N$ is 
a submatrix of $M_G$.  

Therefore, given $M_P$ and $M_N$, we can find a solution
to $\CiPM(M_P, M_N)$ or reject $G$ if no solution exists.
If a solution exists, let $M$ be the implied
taut matrix.  We show that $M$ is a model of $G - N_S$, and, for each
$x \in N_S$, there is a column of $M$ whose probe set is
$N(x)$; we add $x$ to such a column.   We first show that
this is a probe interval model for $G$, and then
show that the model is a normal one, which gives
an $O(n+m)$ bound on the number of 1's in the resulting
matrix.  The $O(n+m)$ time bound therefore follows from 
Lemma~\ref{lem:C1PMTime}.

\begin{definition}\label{def:MPcols}
Let ${\cal S} = \{N(x) \mid x \in N_S$ and $N(x) \neq C[P]$ for any
column $C$ of $M_N\}$.
Let $M_P$ be a matrix obtained by adding each set of ${\cal S}$
as a new column to $M_N[P]$.
\end{definition}

\begin{lemma}\label{lem:MP}
If $M_G$ is a normal model of $G$, 
then some consecutive-ones ordering of $M_P$ is a submatrix of $M_G[P]$. 
\end{lemma}
\begin{proof}
The necessity of columns in ${\cal S}$ in $M_G[P]$ is self-evident, and the necessity
of the remaining columns of $M_P$, which are given by $M_N[P]$ follows from 
Theorem~\ref{thm:invariantCols}. $M_G$ is consecutive-ones ordered, and every 
submatrix of a consecutive-ones ordered matrix is consecutive-ones ordered.
\end{proof}

\begin{lemma}\label{lem:M}
If $\CiPM(M_P,M_N)$ 
has no solution, then $G$ is not a probe interval graph.  Otherwise,
the taut matrix $M$ implied by a solution is a probe interval
model of $G - N_S$.
\end{lemma}
\begin{proof}
Suppose $G$ is a probe interval graph.
Let $M_G$ be a normal model of $G$.  By Theorem~\ref{thm:invariantCols} and Lemma~\ref{lem:MP},
$M_G[P]$ contains a submatrix that is a solution to the consecutive-ones
probe matrix problem on $M_P$ and $M_N$.  Therefore, we can reject
$G$ if this problem has no solution.

Otherwise, let $\pi$ be a solution, and let $M^*_P$ and $M^*_N$ be the consecutive-ones
ordering of $M_P$ and $M_N$ that it gives.
Since $M_N$ is a model of $G - N_S$, so is $M^*_N$.
Let $c$ be a simplicial column of $M_P$ that $\pi$ places between
two neighboring columns $c_1$ and $c_2$ of $M^*_N$.
Since $c$ is not a clique column, the probe set of $c$ is a subset
of $c_1$ or the probe set of $c_2$, by Lemma~\ref{lem:nextProbes}.
Therefore, its probe set is a subset of the neighborhood
of every $v \in V \setminus N_S$ that contains $c$ in $M$,
and its inclusion n $M$ does not alter the represented
neighborhood of $v$.
\end{proof}

\begin{algorithm}\label{alg:MG}
Find a normal model of $G$ from $M_P$ and $M_N$, or determine that
$G$ is not a probe interval graph.

\begin{enumerate}
\item If $\CiPM(M_P,M_N)$ has no solution, return that $G$ is not
a probe interval graph.

\item Otherwise, let $M$ be the taut matrix implied by a solution.  Fill
in the $\ast$'s to get a sparse representation of $M$.
This is $M_G[V \setminus N_S]$.

\item For each $x \in N_S$, add a row for $x$, and put a 1 in a column $c$ of $M$
that has $N(x)$ as its probe set, and return the resulting matrix.
This is $M_G$.
\end{enumerate}
\end{algorithm}

\begin{lemma}\label{lem:MGCorrect}
Algorithm~\ref{alg:MG} is correct.
\end{lemma}
\begin{proof}
By Lemma~\ref{lem:M}, if $\CiPM(M_P,M_N)$ has no solution, $G$ is not a probe
interval graph.  Otherwise, the taut matrix $M = M_G[V \setminus N_S]$ implied by a solution
is a model of $G - N_S$.
The only other adjacencies that need to be represented are between
members of $N_S$ and members of $P$.
By Definition~\ref{def:MPcols}, for each $x \in N_S$, there exists a 
column $c$ whose probe set is equal to $N(x)$.
Placing $x$ in $c$ correctly represents its neighborhood.
This has no effect on adjacencies between other pairs.  
It follows that doing this for all $x \in N_S$
gives a model $M_G$ of $G$.

To see that this is a normal model, observe
that simplicial columns have
different probe sets from all other columns, by Definition~\ref{def:MPcols},
and at least one simplicial non-probe.  Therefore, no simplicial column 
can be merged with a neighboring column.  No non-simplicial column can be merged
with any neighboring non-simplicial column, since these columns
contain a submatrix $M^*_N$ that is a consecutive-ones ordering of $M_N$,
which is a normal model by Theorem~\ref{thm:allModels}.  
Every simplicial non-probe is taut since it occupies only one column.
Every endpoint of a member of $N_1 \cup N_2$ is taut in the taut
matrix $M$; since $M$ is taut, each element of $N_1 \cup N_2$
continues to have its endpoint in a column of $M^*_N$, where it was
taut, since $M^*_N$ is a normal model.
The same argument applies for every endpoint
of a probe that remains in a non-simplicial column.
An endpoint of a probe in a simplicial column is taut, because
the column contains a simplicial non-probe neighbor that resides
only in that column.  The model is a normal one.
\end{proof}

\begin{lemma}\label{lem:MGTime}
Algorithm~\ref{alg:MG} takes $O(n+m)$ time.
\end{lemma}
\begin{proof}
$M_N$ has $O(n+m)$ 1's, by Lemma~\ref{lem:boundedOnes},
because it is a normal model of an induced subgraph of $G$.
Every column of $M_P$ is either a column of $M_N[P]$,
or $N(x)$ for some $x \in N_S$, so the number of 1's in it
is bounded by the size of $M_N$ plus the sum of degrees
of vertices in $N_S$, hence $M_P$ has $O(n+m)$ 1's.
If the algorithm rejects $G$, it therefore takes $O(n+m)$ time
to do so by Lemma~\ref{lem:C1PMTime}.  Otherwise, since a solution
$M$ to this problem is a submatrix of $M_G$, it has $O(n+m)$ 1's,
since $M_G$ is a normal model of $G$, by Lemma~\ref{lem:boundedOnes},
and takes $O(n+m)$ time to produce using elementary sparse matrix
operations.  In addition, for every $x \in N_S$, we identified
a column whose probe set was $N(x)$ when we created $M_P$.
Adding each such $x$ to such a column takes $O(n+m)$ time.
\end{proof}

Summarizing, we obtain the following, which is the main result of the
paper.

\begin{theorem}\label{thm:recognition}
Given a graph $G$ and a partition $\{P,N\}$ of its vertices, 
where $N$ is an independent set,
it takes $O(n+m)$ time to determine whether there is a probe
interval representation of $G$ where $P$ is the set of probes and $N$
is the set of non-probes, and to construct such a representation if it exists.
\end{theorem}

\section{Uniqueness of the model}

If a disconnected probe interval graph has more than two columns in a normal model,
it does not have a unique normal model,
since the order among the components is not constrained,
and columns in one component do not constrain the orderings
of columns in other components.
If the graph has only two components, each with a single column,
then the model is unique up to reversal.  Henceforth, we may resume
our assumption that $G$ is connected.

\begin{algorithm}\label{alg:unique}
Test whether the model $M_G$ returned by Algorithm~\ref{alg:MG} for a connected
graph $G$ is the unique normal model of $G$.

\begin{enumerate}
\item\label{testA} If Algorithm~\ref{alg:C1PMUnique} determines
that $\CiPM(M_P, M_N)$ does not have a unique solution, return false.

\item\label{testB} Else, if $M_G$ has only two columns and one of them
is a simplicial column with two simplicial non-probes, return false.

\item\label{testC}  Else, if $M_G$ has at least three columns, let $T_2$ be
$T(M_P)$ as in Algorithm~\ref{alg:C1PM}. If some non-clique column of $M_G$ that contains a simplicial
non-probe is a child of a P node in $T_2$, return false.

\item  Else return true.
\end{enumerate}
\end{algorithm}

\begin{lemma}\label{lem:MGUSufficient}
If $G$ does not have a unique normal model, Algorithm~\ref{alg:unique}
reports this.
\end{lemma}
\begin{proof}
Let $M_G$ be the model returned by Algorithm~\ref{alg:MG}.

If $M_G$ has only one column, then either the column is a clique column or
the graph contains a single vertex, which is a non-probe. In both cases, it is easy to verify
that this is the only normal model of $G$.

Suppose $M_G$ has two columns. Neither can be a semi-clique column, since
such a column could not have both proper left endpoints and proper right endpoints.
At least one of the two columns, denote it by $k$, must be a clique column, since
otherwise there are no probes in the graph, and the graph is not connected.
If the other column is also a clique column, then
the two cliques of $G$ define the same two clique columns in every normal model of $G$.
As in the single-column case, it is easy to see that any other column in any
model of $G$ can be merged into one of these two clique columns. Therefore, $M_G$
is a unique normal model.

Now consider the case where the other column is a simplicial column and denote
it by $b$. If $b$ contains more than a single simplicial non-probe, then $G$ fails Test~\ref{testB}.
Suppose that $b$ contains a single simplicial non-probe, $x$. In any normal
model of $G$ there must be a clique column equal to $k$ and a simplicial column
equal to $b$ and that contains $N[x]$. Any other column can be merged into one
of these two columns. Therefore, $M_G$ is a unique model in this case as well.

Henceforth, assume that $M_G$ has at least three columns.
We assume that $\CiPM(M_P, M_N)$ has a unique solution,
since otherwise Test~\ref{testA} fails.
Let $A_G$ be a normal model of $G$ that is not equivalent to $M_G$.
For every simplicial column $b$ of $M_G$,
choose a representative simplicial non-probe $y$ in $b$.  
Since $y$ is a simplicial non-probe, it is contained in a unique column 
$b_A$ in $A_G$.  Let us say that we have {\em mapped} $b$ to $b_A$.
Note that $S(b) \cap P = S(b_A) \cap P = N(y)$, and, by the construction of 
$M_G$, $S(c) \cap P \neq N(y)$ for any non-simplicial column $c$. 
By Theorem~\ref{thm:invariantCols}, the sets of vertices represented
by clique and semi-clique columns of $A_G[V \setminus N_S]$ are the same
as the sets of vertices represented by clique and semi-clique columns of
$M_G[V \setminus N_S]$, so $S(c') \cap P \neq N(y)$ for any non-simplicial
column $c'$ of $A_G$.  It follows that $b_A$ is a simplicial column 
of $A_G$.

Let $Y$ be the set of columns of $A_G$ that are either clique columns, 
semi-clique columns, or simplicial columns of $A_G$ to which we have mapped
simplicial columns of $M_G$.
If $Y$ is the entire set of columns of $A_G$, then the set of columns of $A_G[V \setminus N_S]$
is identical to the set of columns of $M_G[V \setminus N_S]$. Since $G$ passed Test~\ref{testA},
the matrices $A_G[V \setminus N_S]$ and $M_G[V \setminus N_S]$ are identical.
Therefore, the only possible difference between $M_G$ and $A_G$ is the
assignment of simplicial non-probes to columns. There is a simplicial non-probe $x$
such that for two different columns $c$ and $c'$, $S(c) \cap P =
S(c') \cap P = N(x)$.  This means that there are two identical columns
in $M_P$, one of which corresponds to the column to which $x$ belongs
in $M_G$. Since the two columns are identical in $M_P$, the sets they represent are the children
of the same P node in $T_2$. By Lemma~\ref{lem:cliquesUnique},
$c$ and $c'$ cannot be clique columns. Therefore $G$ fails Test~\ref{testC}.

Henceforth, assume that there is a column $c$ of $A_G$ that is not in $Y$.
The column $c$ must be a simplicial column, since clique columns and semi-clique
columns are all in $Y$. There is a simplicial non-probe $x$ such that $S(c) \cap P
= N(x)$. There must also be a column $c'$ in $Y$ with
$S(c') \cap P = N(x)$, corresponding to the column of $M_G$ containing $x$.
By Lemma~\ref{lem:cliquesUnique}, $c'$ is not a clique column.
The columns $c$ and $c'$ cannot be consecutive in $A_G$, since otherwise we can merge
them, contradicting the normality of $A_G$. We get that $A_G[P][Y]$ and
$A_G[P][(Y \setminus \{c'\}) \cup \{c\}]$ are two consecutive one orderings of $M_P$ that differ
only in the location of the column corresponding to the one containing $N(x)$.
Therefore, this column is a child of a P node and $G$ fails Test~\ref{testC}.
\end{proof}

To show the converse of this lemma, we first need the following auxiliary
lemma.

\begin{lemma}\label{lem:PNewChild}
Let $\{c\}$ be a child of a P node $B$ in the PQ tree $T$ of a consecutive-ones
matrix $M$, and suppose there is no row $q$ in $M$ such that $S(q) = \{c\}$.
Let $M'$ be the result of adding a new column $c'$
that is a duplicate of $c$.   Then the PQ tree of $M'$ is obtained
from the PQ tree of $M$ by adding $\{c'\}$ as an element of each proper ancestor
of $\{c\}$ and then adding $\{c'\}$ as a new child of $B \cup \{c'\}$.
\end{lemma}

\begin{proof}
Let $C$ be the set of columns of $M$.
Inserting $c'$ next to $c$ in any consecutive-ones ordering of $M$
gives a consecutive-ones ordering of $M'$.  Therefore, $M'$ is a consecutive-ones
matrix and $T(M) \preceq T(M')[C]$.  By Lemma~\ref{lem:RestrictPrec},
$T(M')[C] \preceq T(M)$.  We conclude that $T(M) \equiv T(M')[C]$.
It follows that $T(M)$ is the result of deleting leaf $\{c'\}$ from $T(M')$,
removing it from each of its ancestors,
and contracting its parent if the parent has only one child.  

This contraction is thus the only way that the lemma could fail to be true, and
it occurs if and only if $\{c,c'\}$ is a node of $T(M')$ and $B \cup \{c'\}$ is its
P-node parent.
Suppose this is the case.  Then for a sibling $D$
of $\{c,c'\}$,  there is a row $q'$ of $M'$ such that $S(q')$ contains $c$ and $c'$
but does not contain every member of $D$.  Otherwise, $c$ and $c'$ could be placed
on opposite sides of $D$ in a consecutive-ones ordering, which is forbidden by $\{c,c'\}$.
If $S(q')$ contains any columns other than $c$ and $c'$, then since $S(c)$ does not
contain $D$, not every union of children of $B \cup \{c'\}$ can be a member of $C^{\perp}(M')$,
contradicting Lemma~\ref{lem:CPerp}.   Therefore, $S(q') = \{c,c'\}$.
But then if $q$ is the corresponding row of $M$, $S(q)$ contains
only $c$, contradicting the definition of $c$.
\end{proof}

\begin{lemma}\label{lem:MGUNecessity}
If $G$ fails one of the tests of Algorithm~\ref{alg:unique}, then it does not have a unique normal model.
\end{lemma}
\begin{proof}
For each of the three tests, we show how to construct two non-equivalent models if $G$ fails the test.
Let $M_G$ be the model returned by Algorithm~\ref{alg:MG}.

If $G$ fails Test~\ref{testA}, then $\CiPM(M_P, M_N)$ has two non-equivalent solutions.
Algorithm~\ref{alg:MG} can use each of them to produce a normal model of $G$,
and the two models are not equivalent.

Suppose $G$ fails Test~\ref{testB}. In this case $M_G$ has one clique column $k$
and one simplicial column $b$, which contains at least two simplicial non-probes.
Create a copy $b'$ of $b$. Remove a simplicial non-probe $x$ from $b'$ and all simplicial non-probes other than $x$ from $b$.
Order the three columns $(b, k, b')$.
This is a normal model that is not equivalent to $M_G$.

Suppose that $G$ fails Test~\ref{testC}.
Let $T_2$ be $T(M_P)$. For some simplicial non-probe $x$ that belongs to a non-clique column $c$
of $M_G$, $\{c\}$ is a child of a P node, $B$,  in $T_2$.
The column ordering of $M_G$ gives a consecutive-ones ordering of $M_P$.
Let $D$ be an adjacent sibling to $\{c\}$ in this ordering,
and let $d$ be the column of $D$ that is adjacent to $c$.  Note that $S(d) \neq S(c)$,
otherwise, $M_G$ would have two consecutive columns with the same probe set, and they
could be merged.

Every probe in $c$ also occurs in some other column; otherwise, $c$ would be a clique column.
Add to $M_P$ a new column $c'$ such that $S(c') = S(c)$. Let $T'_2$ be $T(M'_P)$.
By Lemma~\ref{lem:PNewChild}, 
$\{c\}$, $D$, and $\{c'\}$ are children of $B \cup \{c'\}$ in $T'_2$, so $\{c'\}$
can be placed on the opposite side of $D$ from $\{c\}$ in a consecutive-ones ordering of $M_P$.  
Since $S(c') = S(c) \neq S(d)$ and $c$ and $c'$ can be
placed on opposite sides of $d$, $S(c') \subset S(d)$.

Using Algorithm~\ref{alg:MG}, fill in rows for the non-simplicial non-probes
to this matrix, and for every $x' \in N_S$ add a row for $x'$ to $M'$
and put a 1 in the same column that contains 1 in the row of $x'$ in $M_G$. 
Then move the 1 of the simplicial non-probe $x$ from column $c$ to column $c'$.
If $c$ was a simplicial column and no simplicial non-probe remains in it, delete $c$, since it is no
longer a simplicial column.
Apply the procedure from Lemma~\ref{lem:normalModel} to get a normal model $A_G$
of $G$ from $A$.  It may be that $c'$ is merged with other columns.  
However, $c'$ cannot be merged
with $d$ since it contains a simplicial non-probe and $S(c') \cap P \subset S(d) \cap P$.

If $c$ remains, then since every column in a normal model contains a degenerate
or nondegenerate right endpoint, let $y$ be a vertex with a right endpoint in $c$.
In $M_G$, $x$ shares a column with $y$, but not in $A_G$, so $M_G$
and $A_G$ are not equivalent.
If $c$ does not remain and $A_G$ has fewer columns than $M_G$ does, the two
models are not equivalent.
Otherwise, $c$ does not remain and no columns were merged in transforming $A$ to $A_G$,
hence $c'$ remains.
Therefore, $A_G$ is identical to $M_G$ except for the simplicial column
containing $x$, which occurs in a different position in $A_G$.  Since $M_G$ has at 
least three columns, the models are not equivalent.
\end{proof}

Combining the result of this section with Theorem~\ref{thm:recognition} we get the following:

\begin{theorem}
  Given a probe interval graph $G$, in $O(n + m)$ time we can determine whether it has a unique normal model.
\end{theorem}

We note that an implementation of Algorithm~\ref{alg:unique} can be simplified for a two-columns model.
In this case it is enough to apply only Test~\ref{testB}, since $G$ cannot fail Test~\ref{testA}. Also, in
Test~\ref{testC}, we can replace non-clique column with semi-clique column, since if a similicial column
is a child of a P node in $T_2$, $G$ fails Test~\ref{testA}.

\bibliographystyle{plain}

\begin{thebibliography}{10}

\bibitem{AHU74}
Aho, A.V., Hopcroft, J.E., Ullman, J.D.:
\newblock The Design and Analysis of Computer Algorithms.
\newblock Addison-Wesley, Reading, Massachusetts (1974)

\bibitem{Ben59}
Benzer, S.:
\newblock On the topology of the genetic fine structure.
\newblock Proc.\ Nat.\ Acad.\ Sci.\ U.S.A. 45, 1607--1620 (1959)

\bibitem{Booth75}
Booth, K.S.:
\newblock PQ-Tree Algorithms. 
\newblock Ph.D. thesis, Department of Computer Science, University of California, Berkeley, CA, 1975.

\bibitem{BL76}
Booth, K.S., Lueker, G.S.:
\newblock Testing for the consecutive ones property, interval graphs, and graph
  planarity using {PQ}-tree algorithms.
\newblock J.\ Comput.\ Syst.\ Sci. 13, 335--379 (1976)

\bibitem{CGKN07}
Chandler, D.B., Guo, J.,Kloks, T., Niedermeier, R.:
\newblock Probe matrix problems: Totally balanced matrices.
\newblock In: Kao, M.-Y., Li, X.-Y.\ (eds.) AAIM 2007.\ LNCS vol.\ 4508, pp.\ 368--377. Springer, Heidelberg (2007)

\bibitem{CKLP05}
Chang, G.J., Kloks, T., Liu, J., Peng, S.-L.:
\newblock The PIGs full monty - a floor show of minimal separators.
\newblock In: Diekert, V., Durand, B.\ (eds.) STACS 2005.\ LNCS vol.\ 3403, pp.\ 521--532. Springer, Heidelberg (2005)

\bibitem{CLRS09}
Cormen, T.H., Leiserson, C.E., Rivest, R.L., Stein, C.:
\newblock Introduction to Algorithms.
\newblock MIT Press, Cambridge, Massachusetts (2009)

\bibitem{Fullwood09}
Fullwood, M.J., Wei, C.-L., Edison, T.L. et al.:
\newblock Next-generation DNA sequencing of paired-end tags (PET) for
transcriptome and genome analyses.
\newblock Genome Res. 19, 521--532 (2009).

\bibitem{FulkGross65}
Fulkerson, D.R., and Gross, O.:
\newblock Incidence matrices and interval graphs.
\newblock Pacific J.\ Math. 15, 835--855 (1965)

\bibitem{Gol80}
Golumbic, M.C.:
\newblock Algorithmic Graph Theory and Perfect Graphs.
\newblock Academic Press, New York (1980)

\bibitem{GolumbicSandwich98}
Golumbic, M.C.:
\newblock Matrix sandwich problems.
\newblock Linear Algebra and Applications 277, 239--251 (1998)

\bibitem{GolT04}
Golumbic, M.C, Trenk, A.N.:
\newblock Tolerance graphs.
\newblock Cambridge studies in advanced mathematics 89, New York (2004)

\bibitem{JohnsonSpin01}
Johnson, J.L., Spinrad, J.P.:
\newblock A polynomial time recognition algorithm for probe interval graphs.
\newblock In: SODA 2001, pp\ 477--486. Association for Computing Machinery, New York (2001)

\bibitem{LB62}
Lekkerker, C. and Boland, D.:
\newblock Representation of finite graphs by a set of intervals on the real line.
\newblock Fund. Math. 37, 45--64 (1962).

\bibitem{McCPQAlg}
McConnell, R.M., de Montgolfier, F.:
\newblock Algebraic operations on PQ trees and modular decomposition trees.
\newblock In: Kratsch, D.\ (ed.) WG 2005.\ LNCS vol.\ 3787, pp.\ 421--432. Springer, Heidelberg (2005)

\bibitem{McCC1Cert}
McConnell, R.M.:
\newblock A certifying algorithm for the consecutive-ones property.
\newblock In:  SODA 2004, pp.\ 761-770.  Association for Computing Machinery, New York (2004).

\bibitem{McCCircPaper}
McConnell, R.M.:
\newblock Linear-time recognition of circular-arc graphs.
\newblock Algorithmica 37, 93--147 (2003)

\bibitem{McCSpin02}
McConnell, R.M., Spinrad, J.P,:
\newblock Construction of probe interval models.
\newblock In: SODA 2002, pp.\ 866--875. Association for Computing Machinery, New York (2002)

\bibitem{McCNuss09}
McConnell, R.M., Nussbaum, Y.:
\newblock Linear-time recognition of probe interval graphs.
\newblock In:  European Symposium on Algorithms 2009, pp.\ 41-53 (2009)

\bibitem{McMorris98}
McMorris, F.R., Wang, C., Zhang P.,
\newblock On probe interval graphs.
\newblock Discrete Applied Mathematics 88, 315--324 (1998)

\bibitem{MPT98}
Meidanis, J., Porto, O., Telles, G.P.:
\newblock On the consecutive ones property.
\newblock Discrete Applied Mathematics, 88, 325–-354 (1998)

\bibitem{Pevzner00}
Pevzner, P.:
\newblock Computational Molecular Biology: an Algorithmic Approach.
\newblock The MIT Press, Cambridge Massachusetts, 2000.

\bibitem{RTL:triangulated}
Rose, D., Tarjan, R.E., Lueker, G.S.:
\newblock Algorithmic aspects of vertex elimination on graphs.
\newblock SIAM J.\ Comput. 5, 266--283 (1976)

\bibitem{scarpa13}
Donmez, N., Brudno, M.:
\newblock SCARPA:  scaffolding reads with practical algorithms.
\newblock Bioinformatics, 29, 428--434 (2013)

\bibitem{Tucker80}
Tucker, A.C.:
\newblock An efficient test for circular-arc graphs.
\newblock SIAM J.\ Comput.\ 9(1):1--24 (1980)

\bibitem{U04}
Uehara, R.:
\newblock Canonical data structure for interval probe graphs.
\newblock In Fleischer, R., Trippen, G.\ (eds.) ISAAC 2004.\ LNCS vol.\ 3341, pp.\ 859--870. Springer, Heidelberg (2004)

\bibitem{Zhang94}
Zhang, P.:
\newblock Probe interval graph and its applications to physical mapping of DNA.
\newblock Manuscript
\newblock 1994.

\bibitem{Zhang94etal}
Zhang, P. {\em et al.}:
\newblock An algorithm based on graph theory for the assembly of contigs in physical mapping of DNA.
\newblock CABIOS 10, 309--317 (1994).

\bibitem{ZhangPatent}
Zhang, P.:
\newblock United states patent 5667970: Method of mapping {DNA} fragments.
\newblock (July 3, 2000)

\end{thebibliography}

\end{document}